\documentclass[11pt, DIV=15]{scrarticle}
\usepackage[utf8]{luainputenc}
\usepackage[T1]{fontenc}
\usepackage[ngerman,english]{babel}
\usepackage{amsmath, amsfonts, amsthm, amssymb, mathtools, stmaryrd}

\usepackage{algorithm2e}

\usepackage[pdfusetitle,hidelinks]{hyperref}
\usepackage[capitalise,noabbrev]{cleveref}

\usepackage{csquotes}
\usepackage[osf,sc]{mathpazo}
\addtokomafont{disposition}{\normalfont\bfseries} \usepackage{relsize}
\usepackage{microtype}

\usepackage{thm-restate}
\usepackage{subcaption}
\usepackage{enumitem}

\recalctypearea

\usepackage{orcidlink}
\usepackage{booktabs}

\theoremstyle{definition}
\newtheorem{definition}{Definition}[section]

\newtheorem{example}[definition]{Example}

\theoremstyle{plain}
\newtheorem{lemma}[definition]{Lemma}

\newtheorem{theorem}[definition]{Theorem}

\newtheorem{question}[definition]{Question}

\newtheorem{fact}[definition]{Fact}
\Crefname{fact}{Fact}{Facts}

\theoremstyle{remark}

\Crefname{claim}{Claim}{Claims}

\setlist[enumerate, 1]{font=\upshape, noitemsep, nolistsep}
\setlist[enumerate, 2]{font=\upshape, noitemsep, nolistsep}
\setlist[itemize, 1]{noitemsep, nolistsep,font=\upshape}
\setlist[itemize, 2]{noitemsep, nolistsep,font=\upshape}

\DeclareMathOperator\soe{soe}

\DeclareMathOperator\tw{tw}

\DeclareMathOperator\pw{pw}

\DeclareMathOperator\val{val}

\DeclareMathOperator\depth{depth}
\newcommand{\CMSO}{\ensuremath{\mathsf{CMSO}_2}\xspace}
\newcommand{\FO}{\ensuremath{\mathsf{FO}}\xspace}

\usepackage{xspace}
\newcommand{\HomInd}{\textup{\textsc{Hom\-Ind}}\xspace}
\newcommand{\HomIndP}[1]{\textup{\textsc{Hom\-Ind}}\textup{(#1)}\xspace}

\newsavebox{\fminibox}
\newlength{\fminilength}
\newenvironment{fminipage}[1][\linewidth]{
	\setlength{\fminilength}{#1-2\fboxsep-2\fboxrule}\begin{lrbox}{\fminibox}\begin{minipage}{\fminilength}}{
	\end{minipage}\end{lrbox}\noindent\fbox{\usebox{\fminibox}}
}
\newcommand{\dproblem}[3]{
	\begin{center}
		\begin{fminipage}[.95\linewidth]
			\textup{\textsc{#1}}\\
						\textsc{Input:} #2\\
						\textsc{Question:} #3
		\end{fminipage}
\end{center}}
\newcommand{\pproblem}[4]{
	\begin{center}
		\begin{fminipage}[.95\linewidth]
			\textup{\textsc{#1}}\\
			\textsc{Input:} #2\\
			\textsc{Parameter:} #3 \\
			\textsc{Question:} #4
		\end{fminipage}
\end{center}}

\tikzset{
	vertex/.style={draw,circle,fill=gray},
	every node/.style={anchor=center},
	lbl/.style={color=lightgray}
}

\title{An Algorithmic Meta Theorem for Homomorphism Indistinguishability}
\author{Tim Seppelt \orcidlink{0000-0002-6447-0568} \\ \small RWTH Aachen University \\\small
	\href{mailto:seppelt@cs.rwth-aachen.de}{\texttt{seppelt@cs.rwth-aachen.de}}}

\renewcommand{\phi}{\varphi}
\renewcommand{\epsilon}{\varepsilon}

\begin{document}
	\maketitle
	
	\begin{abstract}
		Two graphs $G$ and $H$ are \emph{homomorphism indistinguishable} over a family of graphs $\mathcal{F}$ if for all graphs $F \in \mathcal{F}$ the number of homomorphisms from $F$ to $G$ is equal to the number of homomorphism from $F$ to~$H$. 
		Many natural equivalence relations comparing graphs such as (quantum) isomorphism, cospectrality, and logical equivalences can be characterised as homomorphism indistinguishability relations over various graph classes.

		For a fixed graph class $\mathcal{F}$,
		the decision problem \HomIndP{$\mathcal{F}$} asks to determine whether two input graphs $G$ and $H$ are homomorphism indistinguishable over $\mathcal{F}$.
		The problem \HomIndP{$\mathcal{F}$} is known to be decidable only for few graph classes~$\mathcal{F}$.
		We show that \HomIndP{$\mathcal{F}$} admits a randomised polynomial-time algorithm for every graph class~$\mathcal{F}$ of bounded treewidth which is definable in counting monadic second-order logic \CMSO.
		Thereby, we give the first general algorithm for deciding homomorphism indistinguishability.
		
		This result extends to a version of \HomInd where the graph class $\mathcal{F}$ is specified by a \CMSO-sentence and a bound~$k$ on the treewidth, which are given as input.
		For fixed~$k$, this problem is randomised fixed-parameter tractable.
		If $k$ is part of the input then it is coNP- and coW[1]-hard.
		Addressing a problem posed by Berkholz~(2012),
		we show coNP-hardness by establishing that deciding indistinguishability under the $k$-dimensional Weisfeiler--Leman algorithm is coNP-hard when $k$ is part of the input.
	\end{abstract}

	\section{Introduction}
	
	In 1967, Lovász \cite{lovasz_operations_1967} proved that two graphs $G$ and $H$ are isomorphic if and only if they are \emph{homomorphism indistinguishable} over all graphs, i.e.\ they admit the same number of homomorphisms from every graph $F$.
	Subsequently, many graph isomorphism relaxations have been characterised as homomorphism indistinguishability relations.
	For example, two graphs are quantum isomorphic if and only if they are homomorphism indistinguishable over all planar graphs \cite{mancinska_quantum_2020}.
	Moreover, two graphs satisfy the same sentences in $k$-variable first-order logic with counting quantifiers if and only if they are homomorphism indistinguishable over all graphs of treewidth less than~$k$ \cite{dvorak_recognizing_2010,dell_lovasz_2018}.
A substantial list of similar results characterises notions from quantum information~\cite{mancinska_quantum_2020,atserias_quantum_2019}, finite model theory~\cite{dvorak_recognizing_2010,grohe_counting_2020,fluck_going_2024}, convex optimisation~\cite{dell_lovasz_2018,grohe_homomorphism_2022,roberson_lasserre_2023}, 
	algebraic graph theory~\cite{dell_lovasz_2018,grohe_homomorphism_2022},
	machine learning~\cite{morris_weisfeiler_2019,xu_how_2019,grohe_logic_2021}, and category theory~\cite{dawar_lovasz-type_2021,abramsky_discrete_2022} 
	as homomorphism indistinguishability relations.
	
	The wealth of such examples motivates a more principled study of homomorphism indistinguishability \cite{atserias_expressive_2021,roberson_oddomorphisms_2022,seppelt_logical_2023}.
From a computational perspective, the central question on homomorphism indistinguishability concerns the complexity and computability of the following decision problem \cite[Question~9]{roberson_oddomorphisms_2022}:
	
	\dproblem{\HomIndP{$\mathcal{F}$}}{Graphs $G$ and $H$.}{Are $G$ and $H$ homomorphism indistinguishable over $\mathcal{F}$?}
	
	Here, we consider the graph class $\mathcal{F}$ to be fixed and not part of the input.
	The graphs $G$ and $H$ may be arbitrary graphs and do not necessary have to be in $\mathcal{F}$.
	Typically, the graph class $\mathcal{F}$ is infinite.
	Thus, the trivial approach to \HomIndP{$\mathcal{F}$} of checking whether $G$ and $H$ have the same number of homomorphisms from every $F \in \mathcal{F}$ does not even render \HomIndP{$\mathcal{F}$} decidable.
	
	In general, the understanding of \HomIndP{$\mathcal{F}$} is limited and does not go beyond a short list of examples.
	For arguably artificial graph classes $\mathcal{F}$, \HomIndP{$\mathcal{F}$} can be arbitrarily hard \cite{boeker_complexity_2019,seppelt_logical_2023}.
	However, even for natural graph classes $\mathcal{F}$ the complexity-theoretic landscape of \HomIndP{$\mathcal{F}$} is diverse:
	
	For the class $\mathcal{G}$ of all graphs, \HomIndP{$\mathcal{G}$} is graph isomorphism \cite{lovasz_operations_1967}, 
	a problem representing a long standing complexity-theoretic challenge and currently only known  to be in quasi-polynomial time~\cite{babai_graph_2016}.
	For the class $\mathcal{P}$ of all planar graphs, \HomIndP{$\mathcal{P}$} is undecidable~\cite{mancinska_quantum_2020}.
	For the class $\mathcal{K}$ of all cliques,
	\HomIndP{$\mathcal{K}$} is decidable, yet $\mathsf{C}_=\mathsf{P}$-hard~\cite{boeker_complexity_2019}.
	Finally, for the class $\mathcal{TW}_k$ of all graphs of treewidth at most $k$, \HomIndP{$\mathcal{TW}_k$} can be solved in polynomial time using the well-known $k$-dimensional Weisfeiler--Leman algorithm \cite{dvorak_recognizing_2010,cai_optimal_1992}.
	
	In this work, we generalise the last example by establishing that \HomIndP{$\mathcal{F}$} is tractable for every recognisable graph class~$\mathcal{F}$ of bounded treewidth thus giving the first general algorithm for deciding homomorphism indistinguishability.
	
	\begin{theorem}[restate=mainRP, label=] \label{thm:coRP}
		Let $k \geq 1$.
		If $\mathcal{F}$ is a $k$-recognisable graph class of treewidth at most $k-1$
		then \HomIndP{$\mathcal{F}$} is in \textup{coRP}.
	\end{theorem}
	
	Spelled out, \cref{thm:coRP} asserts that there exists a randomised algorithm for \HomIndP{$\mathcal{F}$} which always runs in polynomial time, accepts all \textsmaller{YES}-instances and incorrectly accepts \textsmaller{NO}-instances with probability less than one half.
	Recognisability is a notion arising in the context of Courcelle's theorem \cite{courcelle_monadic_1990} on model checking, cf.\ \cref{def:recog}.
	Courcelle showed that every graph class definable in counting monadic second-order logic \CMSO is recognisable.
	This subsumes graph classes defined by forbidden (induced) subgraphs and minors, and by the Robertson--Seymour Theorem, all minor-closed graph classes.
	
	Thereby, \cref{thm:coRP} applies to e.g.\ the class of graphs of bounded branchwidth, 
	$k$-outerplanar graphs, 
	and the class of trees of bounded degree.
	The theorem also applies to the 
	graph classes arising in homomorphism indistinguishability characterisations for the Lasserre semidefinite programming hierarchy \cite{roberson_lasserre_2023}. 
	In this way, we answer a question raised in~\cite{roberson_lasserre_2023} affirmatively, cf.\ \cref{thm:lasserre}.

	The proof of \cref{thm:coRP} combines Courcelle's graph algebras \cite{courcelle_graph_2012} with the homomorphism tensors from~\cite{mancinska_quantum_2020,grohe_homomorphism_2022}.
	Graph algebras comprise labelled graphs and operations on them such as series and parallel composition.
	Homomorphism tensors keep track of homomorphism counts of labelled graphs.
	We show that recognisability and bounded treewidth guarantee that
	homomorphism tensors yield finite-dimensional representations of suitable graph algebras which certify homomorphism  indistinguishability and are efficiently computable.
	
	Observing that the dimension of the space of homomorphism tensors corresponds to the number of labels of the labelled graphs in the graph algebra which they represent, one might conjecture the following:
	The problem \HomIndP{$\mathcal{F}$} should be decidable if and only if the graph class $\mathcal{F}$ can be generated as family of labelled graphs with a bounded number of labels, cf.\ \cref{question}.
	Indeed, for planar graphs $\mathcal{P}$, the result~\cite{mancinska_quantum_2020} of Man\v{c}inska and Roberson requires an unbounded number of labels and characterises homomorphism indistinguishability over $\mathcal{P}$ as the existence of some infinite-dimensional operator.
	Intuitively, \HomIndP{$\mathcal{P}$} is undecidable as it requires solving an infinite-dimensional system of equations.
	Our \cref{thm:coRP} makes the other direction of this vague argument precise:
	We show that if the number of labels is bounded (e.g.\ the graph class has bounded treewidth) then considering finite-dimensional spaces suffices, rendering the problem tractable.

	\Cref{thm:coRP} can be strengthened in several ways. 
	Firstly, the connection to Courcelle's theorem motivates considering the following parametrised problem:
	
	\pproblem{\HomInd}{Graphs $G$, $H$, a $\mathsf{CMSO}_2$-sentence $\phi$, an integer $k$.}{$|\phi| + k$.}{Are $G$ and $H$ homomorphism indistinguishable over \( \mathcal{F}_{\phi,k} \coloneqq \{\text{graph }F \mid F \models \phi \text{ and } \tw F \leq k -1\}\)?}
	
	The logical sentence $\phi$ allows the graph class to be specified as part of the input.
	We generalise \cref{thm:coRP} as follows:
	
	\begin{theorem}[restate=mainFPT, label=] \label{thm:fptRP}
		There exists a computable function $f \colon \mathbb{N} \to \mathbb{N}$ and a randomised algorithm for \HomInd of runtime $f(|\phi|+k) n^{O(k)}$ for $n \coloneqq \max \{|V(G)|, |V(H)|\}$ which accepts all \textsmaller{YES}-instances and accepts \textsmaller{NO}-instances with probability less than one half.
	\end{theorem}

	While \cref{thm:coRP} does not depend on Courcelle's work beyond the notion of recognisability, \cref{thm:fptRP} uses his result asserting that given a \CMSO-sentence $\phi$ and $k \in \mathbb{N}$ a finite automaton can be computed which accepts width-$k$ tree-decomposed graphs satisfying $\phi$.
	Based on this automaton,
	our algorithm for \cref{thm:fptRP} computes a basis for  finite-dimensional vector spaces of homomorphism tensors associated to each of its states.
	
	\Cref{thm:coRP} is significant in light of the fact that \HomIndP{$\mathcal{F}$} was previously not even known to be decidable for most graph classes~$\mathcal{F}$ in its scope.
	Despite that,
	it would be desirable to eliminate the need for randomness.
	We accomplish this for a large range of graph classes:
	
	\begin{theorem}[restate=mainPW, label=] \label{thm:pw}
		Let $k \geq 1$.
		If $\mathcal{F}$ is a $k$-recognisable graph class of pathwidth at most $k-1$
		then \HomIndP{$\mathcal{F}$} is in polynomial time.
	\end{theorem}
	
	Furthermore, we show similarly to \cref{thm:fptRP} that the analogue of \HomInd with pathwidth instead of treewidth admits a deterministic algorithm of runtime $f(|\phi|+k)n^{O(k)}$, cf.\ \cref{thm:fptPW}.

	Equipped with the parametrised perspective offered by \HomInd, we finally consider lower bounds on the complexity of this problem.
	Firstly, we show that it is coW[1]-hard and that the runtime of the algorithm that yielded \cref{thm:pw} is optimal under the Exponential Time Hypothesis (\textsmaller{ETH}).
	
	\begin{theorem}[restate=mainD, label=] \label{thm:main4}
		\HomInd is \textup{coW[1]}-hard under fpt-reductions.
		Unless \textsmaller{ETH} fails, there is no algorithm for \HomInd that runs in time $f(|\phi| + k)n^{o(|\phi| + k)}$ for any computable function $f \colon \mathbb{N} \to \mathbb{N}$.
	\end{theorem}

	Secondly, we show that, when disregarding the parametrisation, \HomInd is coNP-hard. We do so by showing coNP-hardness of the following problem involving the Weisfeiler--Leman (WL) algorithm.
	\Cref{thm:main3} is a first step towards resolving a problem posed by Berkholz~\cite{berkholz_lower_2012}: Is the decision problem in \cref{thm:main3} EXPTIME-complete?
	
	\begin{theorem}[restate=mainC, label=] \label{thm:main3}
		The problem of deciding given graphs $G$ and $H$ and an integer $k\in \mathbb{N}$ whether $G$ and $H$ are $k$-\textup{WL} indistinguishable is \textup{coNP}-hard under polynomial-time many-one reductions.
	\end{theorem}	
	WL is an important heuristic in graph isomorphism and tightly related to notions in finite model theory~\cite{cai_optimal_1992} and graph neural networks~\cite{xu_how_2019,morris_weisfeiler_2019,grohe_logic_2021}.
	The best known algorithm for \textsc{WL} runs in time $O(k^2n^{k+1}\log n)$ \cite{immerman_describing_1990}, which is exponential when regarding $k$ as part of the input.
	It was shown that when $k$ is fixed then the problem is PTIME-complete~\cite{grohe_equivalence_1999}.
	Establishing lower bounds on the complexity of \textsc{WL} is a challenging problem~\cite{berkholz_lower_2012,berkholz_tight_2017,grohe_compressing_2023}.

	\section{Preliminaries}
	
	All graphs in this work are simple, undirected, and without multiple edges.
	For a graph $G$, we write $V(G)$ for its vertex set and $E(G)$ for its edge set.
	A \emph{homomorphism} $h \colon F \to G$ from a graph $F$ to a graph $G$ is a map $V(F) \to V(G)$ such that $h(u)h(v) \in E(G)$ for all $uv \in E(F)$.
	Write $\hom(F, G)$ from the number of homomorphisms from $F$ to $G$.
	
	Two graph $G$ and $H$ are \emph{homomorphism indistinguishable} over a class of graph $\mathcal{F}$, in symbols $G \equiv_{\mathcal{F}} H$, if $\hom(F, G) = \hom(F, H)$ for all $F \in \mathcal{F}$.
	For an integer $n \geq 2$, $G$ and $H$ are \emph{homomorphism indistinguishable over $\mathcal{F}$ modulo~$n$}, in symbols $G \equiv_{\mathcal{F}}^n H$, if $\hom(F, G) \equiv \hom(F, H) \mod n$ for all $F \in \mathcal{F}$, cf.\@ \cite{lichter_limitations_2024}.

	Let $\log$ denote the logarithm to base~$2$.
	We recall the following consequences of (an explicit version of) the Prime Number Theorem~\cite{rosser_explicit_1941}.
	\begin{fact} \label{fact:primes}
		For $n \geq e^{2000}$,
		the following hold:
		\begin{enumerate}
			\item the number of primes $n < p \leq n^2$ is at least $\frac{n^2}{2 \log n}$,
			\item the product of all primes $n < p \leq n^2$ is at least is $2^{n^2/2}$,
			\item the probability that a random number $n < m \leq n^2$ is prime is $\geq \frac{1}{2 \log n}$.
		\end{enumerate}
	\end{fact}
	\begin{proof}
		Let $\ln$ denote the natural logarithm.
		Let $\pi(n)$ denote the number of primes $\leq n$.
		By \cite[Theorem~30A]{rosser_explicit_1941}, for $n \geq e^{2000}$,
		\[
		\frac{n}{\ln n} < \pi(n) < \frac{n}{\ln n -2}.
		\]
		Thus the desired quantity $\pi(n^2) - \pi(n)$ is at least
		\begin{align*}
			\frac{n^2}{2 \ln n} - \frac{n}{\ln n - 2}
			\geq \frac{n^2}{2 \ln n} - \frac{2n}{\ln n} 
			= \frac{n^2 - 4n}{2\ln n} 
			\geq \frac{n^2 \ln 2}{2 \ln n}.
		\end{align*}
		For the second claim, the product of the primes is $> n^{\frac{n^2}{2 \log n}} = 2^{\frac{n^2}{2 \log n} \log n} = 2^{n^2/2}$.
		For the third claim, $\frac{n^2}{2 \log n} \cdot \frac{1}{n^2 - n} \geq \frac{1}{2 \log n}$.
	\end{proof}

	\subsection{(Bi)labelled Graphs and Homomorphism Tensors}
	\label{sec:bilabelled}
	
	Let $k, \ell \geq 1$. 
	A \emph{$k$-labelled graph} is a tuple $\boldsymbol{F} = (F, \boldsymbol{u})$ where $F$ is a graph and $\boldsymbol{u} \in V(F)^k$ is a $k$-tuple with distinct entries.
	We say $u_i \in V(F)$, the $i$-th entry of $\boldsymbol{u}$, \emph{carries the $i$-th label}.
	Write $\mathcal{G}(k)$ for the class of $k$-labelled graphs.
	A $k$-labelled graph $\boldsymbol{F} = (F, \boldsymbol{u})$ is \emph{distinctly $k$-labelled} if $u_i \neq u_j$ for all $1 \leq i < j \leq k$.
	Write $\mathcal{D}(k)$ for the class of distinctly $k$-labelled graphs.
		
	A \emph{$(k,\ell)$-bilabelled graph} is a tuple $\boldsymbol{F} = (F, \boldsymbol{u}, \boldsymbol{v})$ where $F$ is a graph and $\boldsymbol{u} \in V(F)^k$ and $\boldsymbol{v} \in V(F)^\ell$.
	It is \emph{distinctly $(k,\ell)$-bilabelled} if  $u_i \neq u_j$ for all $1 \leq i < j \leq k$ and  $v_i \neq v_j$ for all $1 \leq i < j \leq \ell$.
	Note that $\boldsymbol{u}$ and $\boldsymbol{v}$ might share entries.
	We say $u_i \in V(F)$ and $v_i \in V(F)$ \emph{carry the $i$-th in-label} and \emph{out-label}, respectively.
	Write $\mathcal{G}(k,\ell)$ for the class of $(k,\ell)$-bilabelled graphs
	and $\mathcal{D}(k,\ell)$ for the class of distinctly $(k,\ell)$-bilabelled graphs.
	
	For a graph $G$, and $\boldsymbol{F} = (F, \boldsymbol{u}) \in \mathcal{G}(k)$ define its \emph{homomorphism tensor}~$\boldsymbol{F}_G \in \mathbb{N}^{V(G)^k}$ whose $\boldsymbol{v}$-th entry is equal to the number of homomorphisms $h \colon F \to G$ such that $h(u_i) = v_i$ for all $i \in [k]$.
	Analogously, for $\boldsymbol{F} \in \mathcal{G}(k,\ell)$, define $\boldsymbol{F}_G \in \mathbb{N}^{V(G)^k \times V(G)^\ell}$.
	As the entries of homomorphism tensors are integral, we can view them as vectors in vector spaces over $\mathbb{R}$ as in \cref{sec:decidable} or over finite fields $\mathbb{F}_p$ as in \cref{sec:ptime}.
	
	As observed in \cite{mancinska_quantum_2020,grohe_homomorphism_2022}, cf.\ \cite{roberson_lasserre_2023}, (bi)labelled graphs and their homomorphism tensors are intriguing due to the following correspondences between combinatorial operations on the former and algebraic operations on the latter:
	
	Dropping labels corresponds to sum-of-entries (soe):
	For $\boldsymbol{F} = (F, \boldsymbol{u}) \in \mathcal{G}(k)$, define $\soe \boldsymbol{F} \coloneqq F$ as the underlying unlabelled graph of~$\boldsymbol{F}$.
	Then for all graphs $G$, $\hom(\soe \boldsymbol{F}, G) = \sum_{\boldsymbol{v} \in V(G)^k} \boldsymbol{F}_G(\boldsymbol{v}) \eqqcolon \soe(\boldsymbol{F}_G)$.

	Gluing corresponds to Schur products: 
	For $\boldsymbol{F} = (F, \boldsymbol{u})$ and $\boldsymbol{F}' = (F', \boldsymbol{u}')$ in $\mathcal{G}(k)$,
	define $\boldsymbol{F} \odot \boldsymbol{F}' \in \mathcal{G}(k)$ as the distinctly $k$-labelled graph obtained by taking the disjoint union of $F$ and $F'$ and placing the $i$-th label at the vertex obtained by merging $u_i$ with $u'_i$ for all $i \in [k]$.
	Then for every graph~$G$ and $\boldsymbol{v} \in V(G)^k$, $(\boldsymbol{F} \odot \boldsymbol{F}')_G(\boldsymbol{v}) = \boldsymbol{F}_G(\boldsymbol{v}) \boldsymbol{F}'_G(\boldsymbol{v}) \eqqcolon (\boldsymbol{F}_G \odot \boldsymbol{F}'_G)(\boldsymbol{v})$.
	One may similarly define the gluing product of two $(k,\ell)$-bilabelled graphs.
	
	Series composition corresponds to matrix products:
	For $\boldsymbol{K} = (K, \boldsymbol{u}, \boldsymbol{v})$ and $\boldsymbol{K}' = (K', \boldsymbol{u}', \boldsymbol{v}')$ in $\mathcal{G}(k,k)$,
	define $\boldsymbol{K} \cdot \boldsymbol{K}' \in \mathcal{G}(k,k)$
	as the bilabelled graph obtained by taking the disjoint union of $K$ and $K'$,
	identifying $v_i$ with $u'_i$ for $i \in [k]$, 
	and placing the $i$-th in-label (out-label) on $u_i$ (on $v'_i$) for $i \in [k]$.
	Then for all graphs $G$ and $\boldsymbol{x}, \boldsymbol{z} \in V(G)^k$,
	\(
		(\boldsymbol{K} \cdot \boldsymbol{K}')_G(\boldsymbol{x}, \boldsymbol{z})
	   =\sum_{\boldsymbol{y} \in V(G)^k} \boldsymbol{K}_G(\boldsymbol{x}, \boldsymbol{y}) \boldsymbol{K}'_G(\boldsymbol{y}, \boldsymbol{z})
	   \eqqcolon (\boldsymbol{K}_G \cdot \boldsymbol{K}'_G)(\boldsymbol{x}, \boldsymbol{z}).
	\)
	One may similarly compose a graph in $\mathcal{G}(k,k)$ with a graph in $\mathcal{G}(k)$ obtaining one in $\mathcal{G}(k)$.
	This operation correspond to the matrix-vector product.

	Note that all operations above restrict from (bi)labelled graphs to distinctly (bi)labelled graphs. For example, the series product of two distinctly bilabelled graphs is distinctly bilabelled.

	Permutation of labels corresponds to permutation of axes:
	Let $\boldsymbol{F} = (F, \boldsymbol{u}, \boldsymbol{v})\in \mathcal{G}(k,\ell)$ and write $\mathfrak{S}_{k+\ell}$ for the symmetric group acting on $k+\ell$ letters.
	For $\sigma \in \mathfrak{S}_{k+\ell}$, write $\boldsymbol{F}^\sigma \coloneqq (F, \boldsymbol{x}, \boldsymbol{y})$ where $\boldsymbol{x}_i \coloneqq (\boldsymbol{uv})_{\sigma(i)}$ and  $\boldsymbol{y}_{j-k} \coloneqq (\boldsymbol{uv})_{\sigma(j)}$ for all $1 \leq i \leq k < j \leq k+\ell$, i.e.\@ $\boldsymbol{F}^\sigma$ is obtained from $\boldsymbol{F}$ by permuting the labels according to $\sigma$.
	Then for all graphs $G$, $\boldsymbol{F}^\sigma_G$ is obtained from $\boldsymbol{F}_G$ by permuting the axes of the tensor according to $\sigma$, cf.\ \cite{roberson_lasserre_2023}.

	\subsection{Labelled Graphs of Bounded Treewidth}
	\label{sec:labelled}
	
	Labelled graphs of bounded treewidth represent the technical foundation of most proofs in this article. 
	Several versions of such have been discussed in previous works \cite{grohe_homomorphism_2022,courcelle_graph_2012,fluck_going_2024}.
	
	Let $F$ be a graph. A \emph{tree decomposition} of $F$ is a pair $(T, \beta)$ where $T$ is a tree and $\beta \colon V(T) \to 2^{V(F)}$ is a map such that
	\begin{enumerate}
		\item the union of the $\beta(t)$ for $t \in V(T)$ is equal to $V(F)$,
		\item for every edge $e \in E(F)$ there exists $t \in V(T)$ such that $e \subseteq \beta(t)$,
		\item for every vertex $u \in V(F)$ the set of vertices $t \in V(T)$ such that $u \in \beta(t)$ is connected in~$T$.
	\end{enumerate}
	The \emph{width} of $(T, \beta)$ is the maximum over all $\lvert \beta(t) \rvert - 1$ for $t \in V(T)$.
	The \emph{treewidth} $\tw F$ of $F$ is the minimum width of a tree decomposition of $F$.
	A \emph{path decomposition} is a tree decomposition $(T, \beta)$ where $T$ is a path. The \emph{pathwidth} $\pw F$ of $F$ is the minimum width of a path decomposition of $F$.
	
	The following \cref{lem:treedec1,lem:treedec2} introduce tree and path decompositions with certain enjoyable properties.

	\begin{lemma}[{\cite[Lemma~8]{bodlaender_partial_1998}}] \label{lem:treedec1}
		Let $k \geq 1$ and $F$ be a graph such that $|V(F)| \geq k$.
		If $\tw F \leq k-1$ then $F$ admits a tree decomposition $(T, \beta)$ such that
		\begin{enumerate}
			\item $|\beta(t)| = k$ for all $t \in V(T)$,\label{tw1}
			\item $|\beta(s) \cap \beta(t)| = k-1$ for all $st \in E(T)$.\label{tw2}
		\end{enumerate}
		If $\pw F \leq k-1$ then $F$ admits a path decomposition $(T, \beta)$ satisfying \cref{tw1,tw2}.
	\end{lemma}

	The tree decompositions in \cref{lem:treedec1} can be rearranged in order for the depth of the decomposition tree to give a bound on the number of vertices in the decomposed graph.

	\begin{lemma} \label{lem:treedec2}
		Let $k \geq 1$ and $F$ be a graph such that $\tw F \leq k-1$ and $|V(F)| \geq k$.
		Then $F$ admits tree decomposition $(T, \beta)$ such that
		\begin{enumerate}
			\item $|\beta(t)| = k$ for all $t \in V(T)$,
			\item $|\beta(s) \cap \beta(t)| = k-1$ for all $st \in E(T)$,
			\item there exists a vertex $r \in V(T)$ such that the out-degree of every vertex in the rooted tree $(T, r)$ is at most $k$.
		\end{enumerate}
	\end{lemma}
	\begin{proof}
		Let $(T, \beta)$ be a tree decomposition satisfying the assertions of \cref{lem:treedec1}. 
		Pick a root $r \in V(T)$ arbitrarily.
		To ensure that the last property holds, the tree decomposition is modified recursively as follows:
		
		By merging vertices, it can be ensured that no two children of $r$ carry the same bag, i.e.\@ that there exist no two children $s_1 \neq s_2$ of $r$ such that $\beta(s_1) = \beta(s_2)$.
		
		For every $v \in \beta(r)$, let $C(v)$ denote the set of all children $t$ of $r$ such that $\beta(r) \setminus \beta(t) = \{v\}$, i.e.\@ $C(v)$ is the set of all children of $r$ whose bags do not contain the vertex $v$.
		The collection $C(v)$ for $v \in \beta(r)$ is a partition of the children of $r$ in at most $k$ parts.
		
		Note that for two distinct children $t_1 \neq t_2$  in the same part $C(v)$ it holds that $|\beta(t_1) \cap \beta(t_2)| = k-1$.
		Rewire the children of $r$ as follows:
		For every $v \in \beta(r)$ with $C(v) \neq \emptyset$, pick $t \in C(v)$, make $t$ a child of $r$ and all other elements of $C(v)$ children of $t$.
		
		The vertex $r$ now has at most $k$ children and the new tree decomposition still satisfies the assertions of \cref{lem:treedec1}. Proceed by processing the children of $r$.
	\end{proof}	

	Inspired by \cref{lem:treedec1}, 
	we consider the following family of distinctly labelled graphs.
	The \emph{depth} of a rooted tree $(T, r)$ is the maximal number of vertices on any path from $r$ to a leaf.
	
	\begin{definition}\label{def:twk}
		Let $k,d \geq 1$.
		Define $\mathcal{TW}_d(k)$ as the set of all $\boldsymbol{F} = (F, \boldsymbol{u}) \in \mathcal{D}(k)$ such that
		$F$ admits a tree decomposition $(T, \beta)$ of width $\leq k-1$ 
		satisfying the assertions of \cref{lem:treedec2} with some $r \in V(T)$
		such that $\beta(r) = \{u_1, \dots, u_k\}$ and $(T, r)$ is of depth at most $d$.
		Let $\mathcal{TW}(k) \coloneqq \bigcup_{d\geq 1} \mathcal{TW}_d(k)$.
	\end{definition}

	Crucially, this definition permits the following bound on the size of the graphs in $\mathcal{TW}_d(k)$ in terms of $d$ and $k$.

	\begin{lemma} \label{lem:tw-size}
		Let $k ,d\geq 1$.
		Every $\boldsymbol{F} \in \mathcal{TW}_d(k)$ has at most $\max\{k^d, d\}$~vertices.
	\end{lemma}
	\begin{proof}
		Let $(T, \beta)$ and $r \in V(T)$ be as in \cref{def:twk}.
		If $k = 1$, then every vertex in $(T,r)$ has out-degree $1$ and $F$ at most $d$ vertices.
		
		Suppose that $k \geq 2$.
		The proof is by induction on the depth $d$ of the rooted tree $(T, r)$.
		If $d = 1$, then $T$ contains only a single vertex and $\boldsymbol{F}$ has at most $k$ vertices.
		
		For the inductive step, let $\boldsymbol{F}$ be of depth $d$.
		If $r$ has only a single neighbour $s$, then $S \coloneqq T - r$ is such that $(S, s)$ is of depth $d-1$.
		By the inductive hypothesis, $\left\lvert \bigcup_{s\in V(S)} \beta(s) \right\rvert \leq k^d$.
		Furthermore, $\left\lvert \bigcup_{t\in V(T)} \beta(t) \setminus \bigcup_{s\in V(S)} \beta(s) \right\rvert = 1$.
		Hence, $\boldsymbol{F}$ has at most $k^{d-1} + 1 \leq k^{d}$ many vertices.
		
		If $r$ has multiple neighbours, observe that due to \cref{lem:treedec2} every vertex in $\beta(r)$ is also in $\beta(s)$ for some neighbour $s$ of $r$.
		Hence, the number of vertices in $\boldsymbol{F}$ is bounded by the number of vertices covered by the subtrees of $T - r$ rooted in $s$.
		Thus, $\boldsymbol{F}$ has at most $k^{d-1} \cdot k \leq k^{d}$ many vertices.
	\end{proof}

	Clearly, if $\boldsymbol{F} \in \mathcal{TW}(k)$ then $\tw (\soe \boldsymbol{F}) \leq k-1$.
	Conversely, by \cref{lem:treedec2}, for every $F$ with $\tw F \leq k-1$ and $\lvert V(F) \rvert \geq k$,
	there exists $\boldsymbol{u} \in V(F)^k$ such that $(F, \boldsymbol{u}) \in \mathcal{TW}(k)$.
	Thus the underlying unlabelled graphs of the labelled graphs in $\mathcal{TW}(k)$ are exactly the graphs of treewidth $\leq k-1$ on $\geq k$ vertices.
	
	The family $\mathcal{TW}(k)$ is generated by certain small building blocks under series composition and gluing as follows.
	
	Let $\boldsymbol{1} \in \mathcal{TW}_1(k)$ be the distinctly $k$-labelled graph on $k$ vertices without any edges.
	For $i \in [k]$, let $\boldsymbol{J}^i = (J^i, (1, \dots, k), \allowbreak (1, \dots, i-1, \widehat{i}, i+1, \dots, k))$ the distinctly $(k,k)$-bilabelled graph with $V(J^i) \coloneqq [k] \cup \{\widehat{i}\}$ and $E(J^i) \coloneqq \emptyset$.
	Writing $\binom{[k]}{2}$ for the set of pairs of distinct elements in $[k]$,
	let $\boldsymbol{A}^{ij} = (A^{ij}, (1, \dots, k), (1, \dots, k))$ for $ij \in \binom{[k]}{2}$ be the distinctly $(k,k)$-bilabelled graph with $V(A^{ij}) \coloneqq [k]$ and $E(A^{ij}) \coloneqq \{ij\}$.
	These graphs are depicted in \cref{fig:basal}.
	Let $\mathcal{B}(k) \coloneqq \{\boldsymbol{J}^i \mid i \in [k]\} \allowbreak \cup \{\boldsymbol{A}^{ij} \mid ij \in \binom{[k]}{2}\}$.
	Observe that $\lvert \mathcal{B}(k) \rvert \leq k^2$.
	
	\begin{figure}
		\begin{subfigure}{.3\linewidth}
			\centering
			\begin{tikzpicture}
				\node (a) [vertex] {};
				\node (ab) [vertex,below of=a] {};
				\node (b) [below of=ab, yshift=.5cm] {\vdots};
				\node (f) [below of=b,vertex, yshift=.3cm] {};
				\node (fb) [below of=f,vertex] {};
				
				\node (a1) [lbl,left of=a] {$1$};
\draw [color=lightgray] (a) -- (a1);

				\node (ab1) [lbl,left of=ab] {$2$};
\draw [color=lightgray] (ab) -- (ab1);
\node (f1) [lbl,left of=f] {$k-1$};
\draw [color=lightgray] (f) -- (f1);
\node (fb1) [lbl,left of=fb] {$k$};
\draw [color=lightgray] (fb) -- (fb1);
\end{tikzpicture}
			\caption{$\boldsymbol{1} \in \mathcal{TW}(k)$.}
			\label{fig:one}
		\end{subfigure}
		\begin{subfigure}{.3\linewidth}
			\centering
			\begin{tikzpicture}
				\node (a) [vertex] {};
				\node (b) [below of=a, yshift=.5cm] {\vdots};
				\node (c) [below of=b,vertex, yshift=.3cm] {};
				\node (d) [below of=c,vertex] {};
				\node (e) [below of=d, yshift=.5cm] {\vdots};
				\node (f) [below of=e,vertex, yshift=.3cm] {};
				\draw [ultra thick] (c) edge [bend left] (d); 
				\node (e) [below of=c, yshift=.6cm] {\vdots};
				
				\node (a1) [lbl,left of=a] {$1$};
				\node (a2) [lbl,right of=a] {$1$};
				\draw [color=lightgray] (a) -- (a1);
				\draw [color=lightgray] (a) -- (a2);
				
				\node (c1) [lbl,left of=c] {$i$};
				\node (c2) [lbl,right of=c] {$i$};
				\draw [color=lightgray] (c) -- (c1);
				\draw [color=lightgray] (c) -- (c2);
				\node (d1) [lbl,left of=d] {$j$};
				\node (d2) [lbl,right of=d] {$j$};
				\draw [color=lightgray] (d) -- (d1);
				\draw [color=lightgray] (d) -- (d2);
				\node (f1) [lbl,left of=f] {$k$};
				\node (f2) [lbl,right of=f] {$k$};
				\draw [color=lightgray] (f) -- (f1);
				\draw [color=lightgray] (f) -- (f2);
			\end{tikzpicture}
			\caption{$\boldsymbol{A}^{ij} \in \mathcal{D}(k,k)$.}
		\end{subfigure}
		\begin{subfigure}{.3\linewidth}
			\centering
			\begin{tikzpicture}
				\node (a) [vertex] {};
				\node (b) [below of=a, yshift=.5cm] {\vdots};
				\node (c) [below of=b,vertex, yshift=.3cm] {};
				\node (dx) [below of=c,vertex, xshift=-.5cm, yshift=.5cm] {};
				\node (dy) [below of=c,vertex, xshift=.5cm, yshift=.5cm] {};
				\node (dz) [below of=dx, xshift=.5cm,vertex, yshift=.5cm] {};
				\node (e) [below of=dz, yshift=.5cm] {\vdots};
				\node (f) [below of=e,vertex, yshift=.3cm] {};

				\node (a1) [lbl,left of=a,xshift=-.5cm] {$1$};
				\node (a2) [lbl,right of=a,xshift=.5cm] {$1$};
				\draw [color=lightgray] (a) -- (a1);
				\draw [color=lightgray] (a) -- (a2);
				
				\node (c1) [lbl,left of=c,xshift=-.5cm] {$i-1$};
				\node (c2) [lbl,right of=c,xshift=.5cm] {$i-1$};
				\draw [color=lightgray] (c) -- (c1);
				\draw [color=lightgray] (c) -- (c2);
				\node (d1) [lbl,left of=dx] {$i$};
				\node (d2) [lbl,right of=dy] {$i$};
				\draw [color=lightgray] (d1) to (dx);
				\draw [color=lightgray] (d2) to (dy);
				
				\node (d3) [lbl,left of=dz,xshift=-.5cm] {$i+1$};
				\node (d4) [lbl,right of=dz,xshift=.5cm] {$i+1$};
				\draw [color=lightgray] (dz) -- (d3);
				\draw [color=lightgray] (dz) -- (d4);
				
				\node (f1) [lbl,left of=f,xshift=-.5cm] {$k$};
				\node (f2) [lbl,right of=f,xshift=.5cm] {$k$};
				\draw [color=lightgray] (f) -- (f1);
				\draw [color=lightgray] (f) -- (f2);
			\end{tikzpicture}
			\caption{$\boldsymbol{J}^{i} \in \mathcal{D}(k,k)$.}
		\end{subfigure}
		\caption{The (bi)labelled generators of $\mathcal{TW}(k)$ in wire notion of \cite{mancinska_quantum_2020}. A vertex carries in-label (out-label) $i$ if it is connected to the index~$i$ on the left (right) by a wire. Actual edges and vertices of the graph are depicted in black.}
		\label{fig:basal}
	\end{figure}
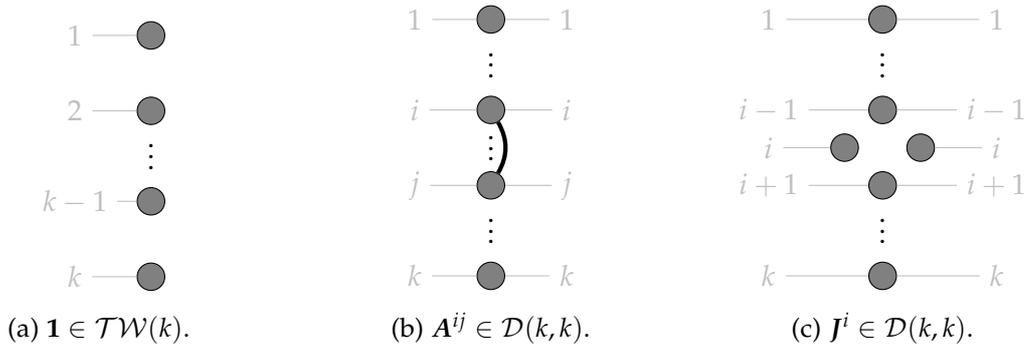

	It can be readily verified that if $\boldsymbol{F} \in \mathcal{TW}(k)$ and $\boldsymbol{B} \in \mathcal{B}(k)$ then $\boldsymbol{B} \cdot \boldsymbol{F} \in \mathcal{TW}(k)$.
	Furthermore, if $\boldsymbol{F}, \boldsymbol{F}' \in \mathcal{TW}(k)$ then $\boldsymbol{F} \odot \boldsymbol{F}' \in \mathcal{TW}(k)$.
	Conversely, the elements of $\mathcal{B}(k)$ generate $\mathcal{TW}(k)$ in the following sense:
	
	\begin{lemma} \label{lem:twk-gen}
		Let $k \geq 1$.
		For every $\boldsymbol{F} \in \mathcal{TW}(k)$, one of the following holds:
		\begin{enumerate}
			\item $\boldsymbol{F} = \boldsymbol{1}$,
			\item $\boldsymbol{F} = \prod_{ij \in A} \boldsymbol{A}^{ij} \cdot \boldsymbol{F}'$ for some $A \subseteq \binom{[k]}{2}$ and $\boldsymbol{F}' \in \mathcal{TW}(k)$ with less edges than $\boldsymbol{F}$,
			\item $\boldsymbol{F} = \boldsymbol{J}^i \cdot \boldsymbol{F}'$ for some $i \in [k]$ and $\boldsymbol{F}' \in \mathcal{TW}(k)$ with less vertices than $\boldsymbol{F}$,
			\item $\boldsymbol{F} = \boldsymbol{F}_1 \odot \boldsymbol{F}_2$ for $\boldsymbol{F}_1, \boldsymbol{F}_2 \in \mathcal{TW}(k)$ on less vertices than~$\boldsymbol{F}$.
		\end{enumerate}
	\end{lemma}
	\begin{proof}
		Write $\boldsymbol{F} = (F, \boldsymbol{u})$ and
		let $(T, \beta)$ and $r \in V(T)$ be for $\boldsymbol{F} \in \mathcal{TW}(k)$ as in \cref{def:twk}.
		
		For the second case, let $A = \{ij \in \binom{[k]}{2} \mid u_iu_j \in E(F)\}$.
		Write $F'$ for the graph obtained from $F$ by deleting the edges $u_iu_j$ for $ij \in A$.
		Then $\boldsymbol{F}' \coloneqq (F', \boldsymbol{u}) \in \mathcal{TW}(k)$
		and $\boldsymbol{F} = \prod_{ij \in A} \boldsymbol{A}^{ij} \cdot \boldsymbol{F}'$.
		
		If $A = \emptyset$, distinguish two cases:
		The vertex $r$ has a unique neighbour $s \in V(T)$.
		By the assertions of \cref{lem:treedec1}, 
		there exists a unique $i \in [k]$ such that $u_i \not\in \beta(s)$.
		Let $F'$ be the subgraph of $F$ induced by $\bigcup_{t \in V(T) \setminus \{r\}} \beta(t)$
		and write $\boldsymbol{v} \in V(F')^k$ for the tuple of distinct vertices which coincides with $\boldsymbol{u}$
		everywhere except at the $i$-th entry where it is equal to the unique vertex in $\beta(s) \setminus \beta(r)$.
		Then $\boldsymbol{F}' \coloneqq (F', \boldsymbol{v}) \in \mathcal{TW}(k)$ and $\boldsymbol{F} = \boldsymbol{J}^i \cdot \boldsymbol{F}'$.
		
		Finally, suppose that $r$ has multiple neighbours $s_1, \dots, s_m \in V(T)$.
		Let $T_1$ denote the connected component of $T - \{s_2, \dots, s_m\}$ containing $s_1$.
		Let $T_2$ denote the connected component of $T - \{s_1\}$ containing $s_2, \dots, s_m$.
		For $j \in \{1,2\}$, let $F_j$ be the subgraph induced by $\bigcup_{t \in V(T_j)} \beta(t)$. 
		Then $F_1$ contains less vertices than $F$ since the vertex in $\beta(s_2) \setminus \beta(r)$ is not in~$F_1$.
		Symmetrically, $F_2$ contains less vertices than $F$.
		Then $\boldsymbol{F}_j \coloneqq (F_j, \boldsymbol{u}) \in \mathcal{TW}(k)$ for $j \in \{1,2\}$.
		As desired, $\boldsymbol{F} = \boldsymbol{F}_1 \odot \boldsymbol{F}_2$.
	\end{proof}

 	\subsection{Labelled Graphs of Bounded Pathwidth}
 	\label{sec:labelled-pw}
 	
 	The set-up from \cref{sec:labelled} can be adapted for pathwidth as follows.
 	
 	\begin{definition} \label{def:pwk}
 		For $k, d \geq 1$,
 		write $\mathcal{PW}_d(k) \subseteq \mathcal{D}(k)$ for the class of all distinctly $k$-labelled graphs $\boldsymbol{F} = (F, \boldsymbol{u})$ such that there exists a path decomposition $(P, \beta)$ of $F$ where
 		\begin{enumerate}
 			\item $|\beta(t)| = k$ for all $t \in V(P)$,
 			\item $|\beta(s) \cap \beta(t)| = k-1$ for all $st \in E(P)$,
 			\item there exists a vertex $r \in V(P)$ of degree at most $1$ such that $\beta(r) = \{u_1, \dots, u_k\}$,
 			\item $|V(P)| \leq d$.
 		\end{enumerate}
 		Let $\mathcal{PW}(k) \coloneqq \bigcup_{d\geq 1} \mathcal{PW}_d(k)$.
 	\end{definition}
 
 	By \cref{lem:treedec1}, for every graph $F$ with $\pw F \leq k-1$ on at least $k$ vertices, there exists $\boldsymbol{u} \in V(F)^k$ such that $(F, \boldsymbol{u}) \in \mathcal{PW}(k)$.
 	Conversely, the underlying graph of every $\boldsymbol{F} \in \mathcal{PW}(k)$ has pathwidth at most $k-1$.
 	Analogously to \cref{lem:tw-size}, we obtain the following bound on the size of the graphs in $\mathcal{PW}_d(k)$:
 	
 	\begin{lemma} \label{lem:pwk-size}
 		Let $k, d \geq 1$.
 		Every $\boldsymbol{F} \in \mathcal{PW}_d(k)$ has at most $k + d- 1$ vertices.
 	\end{lemma}
 	\begin{proof}
 		If $d = 1$, then the path decomposition $(P, \beta)$ of $\boldsymbol{F}$ as in \cref{def:pwk} has at most one bag and $\boldsymbol{F}$ has at most $k$ vertices.
 		For larger $d$, observe that, when traversing the decomposition from one end of $P$ to the other, every bag introduces exactly one new vertex. Thus, there are at most $k+d-1$ vertices in $\boldsymbol{F}$.
 	\end{proof}

	Similar to $\mathcal{TW}(k)$, $\mathcal{PW}(k)$ is generated by small building blocks. For $\mathcal{PW}(k)$, only series composition needs to be considered. Gluing is not permitted.
	Clearly, if $\boldsymbol{F} \in \mathcal{PW}(k)$ then $\boldsymbol{B} \cdot \boldsymbol{F}\in \mathcal{PW}(k)$ for every $\mathcal{B}(k)$.
	\begin{lemma} \label{lem:pwk-gen}
		Let $k \geq 1$.
		For every $\boldsymbol{F} \in \mathcal{PW}(k)$, one of the following holds:
		\begin{enumerate}
			\item $\boldsymbol{F} = \boldsymbol{1}$,
			\item $\boldsymbol{F} = \prod_{ij \in A} \boldsymbol{A}^{ij} \cdot \boldsymbol{F}'$ for some $A \subseteq \binom{[k]}{2}$ and $\boldsymbol{F}' \in \mathcal{PW}(k)$ with less edges than $\boldsymbol{F}$,
			\item $\boldsymbol{F} = \boldsymbol{J}^i \cdot \boldsymbol{F}'$ for some $i \in [k]$ and $\boldsymbol{F}' \in \mathcal{PW}(k)$ with less vertices than $\boldsymbol{F}$.
		\end{enumerate}
	\end{lemma}
	\begin{proof}
		Following the proof of \cref{lem:twk-gen}, observe that the missing case does not occur.
	\end{proof}
			\subsection{Counting Monadic Second-Order Logic}

	We denote \emph{first-order logic} over the relational signature $\{E\}$ of graphs by $\FO$.
	The logic $\mathsf{MSO}_1$ is the extension of \FO by variables $X, Y, Z, \dots$ which range over subsets of vertices and atomic formulas $x \in X$.
	The logic $\mathsf{MSO}_2$ extends $\mathsf{MSO}_1$ by first-order (monadic second-order) variables which range over edges (sets of edges) and atomic formulas $\operatorname{inc}(x, y)$ evaluates to true if $x$ is assigned a vertex $v$ and $y$ is assigned an edge $e$ such that $v$ is incident with $e$.
	\emph{Counting monadic second-order logic} \CMSO is the extension of $\mathsf{MSO}_2$ by atomic formulas $\operatorname{card}_{p,q}(X)$ for integers $p, q \in \mathbb{N}$ expressing $\lvert X \rvert \equiv p \mod q$.
	See \cite{courcelle_graph_2012} for further details.

	\subsection{Recognisability}
	
	We consider the following definition of recognisability \cite{bojanczyk_definability_2016}.
	
	\begin{definition} \label{def:recog}
		Let $k \geq 1$.
		For class of unlabelled graphs $\mathcal{F}$, define the equivalence relation $\sim_{\mathcal{F}}$ on the class of distinctly $k$-labelled graphs~$\mathcal{D}(k)$ by
		letting $\boldsymbol{F}_1 \sim_{\mathcal{F}} \boldsymbol{F}_2$ if and only if for all $\boldsymbol{K} \in \mathcal{D}(k)$ it holds that
		\[
		\soe (\boldsymbol{K} \odot \boldsymbol{F}_1) \in \mathcal{F} \iff \soe (\boldsymbol{K} \odot \boldsymbol{F}_2) \in \mathcal{F}.
		\]
		The class $\mathcal{F}$ is \emph{$k$-recognisable} if $\sim_{\mathcal{F}}$
		has finitely many equivalence classes.
		The number of classes is the \emph{$k$-recognisability index} of $\mathcal{F}$.
	\end{definition}
	
	To parse \cref{def:recog}, first recall that $\boldsymbol{K} \odot \boldsymbol{F}_1$ is the $k$-labelled graph by gluing $\boldsymbol{K}$ and $\boldsymbol{F}_1$ together at their labelled vertices. The $\soe$-operator drops the labels yielding unlabelled graphs.
	Intuitively, $\boldsymbol{F}_1 \sim_{\mathcal{F}} \boldsymbol{F}_2$ iff both or neither of their underlying unlabelled graphs are in $\mathcal{F}$ and the positions of the labels in $\boldsymbol{F}_1$ and $\boldsymbol{F}_2$ is equivalent with respect to membership in $\mathcal{F}$.
	Consider the following example:
	\begin{example} \label{ex:paths}
		The class $\mathcal{W}$ of all paths is $1$-recognisable.
		Its $1$-recognisability index is $4$.
		The equivalence classes are described by the representatives in \cref{fig:paths}.
	\end{example}
	
	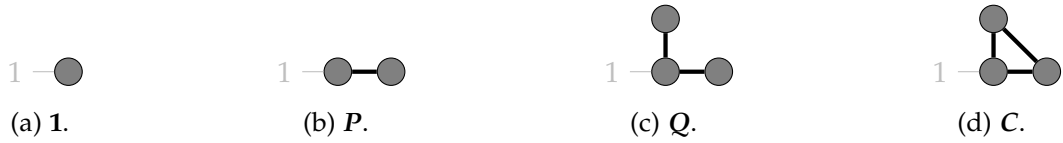
\begin{figure}
		\begin{subfigure}[b]{.2\linewidth}
			\centering
			\begin{tikzpicture}[node distance=.7cm]
				\node (a) [vertex] {};
				\node (a1) [lbl,left of=a] {$1$};
				\draw [color=lightgray] (a) -- (a1);
			\end{tikzpicture}
			\caption{$\boldsymbol{1}$.}
		\end{subfigure}
		\begin{subfigure}[b]{.25\linewidth}
			\centering
			\begin{tikzpicture}[node distance=.7cm]
				\node (a) [vertex] {};
				\node (b) [vertex, right of=a] {};
				\draw [ultra thick] (a) -- (b); 
				\node (a1) [lbl,left of=a] {$1$};
				\draw [color=lightgray] (a) -- (a1);
			\end{tikzpicture}
			\caption{$\boldsymbol{P}$.}
		\end{subfigure}	
		\begin{subfigure}[b]{.25\linewidth}
			\centering
			\begin{tikzpicture}[node distance=.7cm]
				\node (a) [vertex] {};
				\node (b) [vertex, right of=a] {};
				\node (c) [vertex, above of=a] {};
				\draw [ultra thick] (a) -- (b); 
				\draw [ultra thick] (a) -- (c); 
				\node (a1) [lbl,left of=a] {$1$};
				\draw [color=lightgray] (a) -- (a1);
				
			\end{tikzpicture}
			\caption{$\boldsymbol{Q}$.}
		\end{subfigure}	
		\begin{subfigure}[b]{.25\linewidth}
			\centering
			\begin{tikzpicture}[node distance=.7cm]
				\node (a) [vertex] {};
				\node (b) [vertex, right of=a] {};
				\node (c) [vertex, above of=a] {};
				\draw [ultra thick] (a) -- (b); 
				\draw [ultra thick] (a) -- (c); 
				\draw [ultra thick] (b) -- (c);
				\node (a1) [lbl,left of=a] {$1$};
				\draw [color=lightgray] (a) -- (a1);
				
			\end{tikzpicture}
			\caption{$\boldsymbol{C}$.}
		\end{subfigure}		
		\caption{Representatives for $\sim_{\mathcal{W}}$ for $k=1$ and $\mathcal{W}$ the class of paths from \cref{ex:paths}.}
		\label{fig:paths}
	\end{figure}
	
	\begin{proof}
		To show that the labelled graphs in \cref{fig:paths} cover all equivalence classes, let $\boldsymbol{F} = (F, u)$ be arbitrary.
		If $F$ is not a path then $\boldsymbol{F} \sim_{\mathcal{W}} \boldsymbol{C}$. Indeed, for every $\boldsymbol{K} \in \mathcal{D}(1)$, $F$ is a subgraph of $\soe(\boldsymbol{K} \odot \boldsymbol{F})$. Hence, regardless of $\boldsymbol{K}$, both $\soe(\boldsymbol{K} \odot \boldsymbol{F})$ and $\soe(\boldsymbol{K} \odot \boldsymbol{C})$ are not paths.
		If $F$ is a path, then $\boldsymbol{F}$ and $\boldsymbol{1}$, $\boldsymbol{P}$, or $\boldsymbol{Q}$ are equivalent depending on whether the degree of $u$ is $0$, $1$, or $2$.
		
		To show that the representatives in \cref{fig:paths} are in distinct classes, observe for example that $\soe(\boldsymbol{P} \odot \boldsymbol{P}) \in \mathcal{W}$ while $\soe(\boldsymbol{P} \odot \boldsymbol{Q}) \not\in \mathcal{W}$, thus $\boldsymbol{P} \not\sim_{\mathcal{W}} \boldsymbol{Q}$.
		Similarly, $\soe(\boldsymbol{1} \odot \boldsymbol{Q}) \in \mathcal{W}$ whereas $\soe(\boldsymbol{P} \odot \boldsymbol{Q}) \not\in \mathcal{W}$, thus $\boldsymbol{1} \not\sim_{\mathcal{W}} \boldsymbol{P}$.
	\end{proof}
	A more involved example is the following:
	\begin{example}
		Let $\mathcal{F}$ be the family of $H$-subgraph-free graph for some graph $H$.
		Then $\mathcal{F}$ is $k$-recognisable for every $k \geq 1$.
	\end{example}
	\begin{proof}
		Suppose $H$ has $m$ vertices.
		For a distinctly $k$-labelled graph $\boldsymbol{F} = (F, \boldsymbol{u})$, consider the set $\mathcal{H}(\boldsymbol{F})$ of (isomorphism types of) distinctly $k$-labelled graphs $\boldsymbol{F}' = (F', \boldsymbol{u})$ where $F'$ is a subgraph of $F$ such that $V(F') \supseteq \{u_1, \dots, u_k\}$ has at most $k+m$ vertices.
		Clearly, there are only finitely many possible sets $\mathcal{H}(\boldsymbol{F})$.
		Furthermore, if $\mathcal{H}(\boldsymbol{F}_1) = \mathcal{H}(\boldsymbol{F}_2)$ then $\boldsymbol{F}_1 \sim_{\mathcal{F}} \boldsymbol{F}_2$.
		Indeed, if $\boldsymbol{K}  \in \mathcal{D}(k)$ is such that $\soe(\boldsymbol{K} \odot \boldsymbol{F}_1)$ contains $H$ as a subgraph then so does $\soe(\boldsymbol{K} \odot \boldsymbol{F}_2)$ since $\boldsymbol{F}_1$ and $\boldsymbol{F}_2$ contain the same distinctly $k$-labelled graphs on at most $k+m$ vertices.
	\end{proof}
	Analogously, one may argue that every class defined by forbidden minors is recognisable.

	Courcelle~\cite{courcelle_monadic_1990} proved that every \CMSO-definable graph class is \emph{recognisable}, 
	i.e.\@ it is $k$-recognisable for every $k \in \mathbb{N}$.
	Conversely, Boja\'nczyk and Pilipczuk~\cite{bojanczyk_definability_2016} proved that if a recognisable class $\mathcal{F}$ has bounded treewidth then it is \CMSO-definable.
	Furthermore, they conjecture that $k$-recognisability is a sufficient condition for a graph class of treewidth at most $k-1$ to be \CMSO-definable.

	For our purposes, it suffices to weaken \cref{def:recog} as follows:
	The equivalence relation~$\sim_{\mathcal{F}}$ can be defined on $\mathcal{TW}(k)$ instead of $\mathcal{G}(k)$.
	Then $\boldsymbol{K}$ ranges over $\mathcal{TW}(k)$.
	We chose the definition from \cite{courcelle_monadic_1990,bojanczyk_definability_2016} in order to be aligned with the literature.
	
		\subsection{Cai--Fürer--Immerman Graphs}
	\label{sec:cfi}	
	We recall the following version of the classical \textsmaller{CFI} graphs~\cite{cai_optimal_1992} from \cite{roberson_oddomorphisms_2022}.
	Let $G$ be a graph without isolated vertices and $U \colon V(G) \to \mathbb{Z}_2$ a function from $G$ to the abelian group on two elements~$\mathbb{Z}_2$.
	For a vertex $v \in V(G)$, write $E(v) \subseteq E(G)$ for the set of edges incident to~$v$.
	The graph $G_U$ has vertices $(v, S)$ for every $v \in V(G)$ and $S \colon E(v) \to \mathbb{Z}_2$ such that $\sum_{e \in E(v)} S(e) = U(v)$.
	Two vertices $(u, S)$ and $(v, T)$ are adjacent in~$G_U$ if $uv \in E(G)$ and $S(uv) + T(uv) = 0$.
	Note that $\lvert V(G_U)\rvert = \sum_{v \in V(G)}2^{\deg(v)-1}$.
	
	By \cite[Lemma~3.2]{roberson_oddomorphisms_2022}, if $\sum_{v \in V(G)} U(v) = \sum_{v \in V(G)} U'(v)$ for $U, U' \colon V(G) \to \mathbb{Z}_2$ then $G_U = G_{U'}$.
	We may thus write $G_0$ and $G_1$ for the \emph{even} and the \emph{odd} \textsmaller{CFI} graph of $G$.
	We recall the following properties:
	\begin{lemma}[{\cite[Corollary~3.7]{roberson_oddomorphisms_2022}}] \label{lem:cfi}
		Let $G$ be a connected graph and $U \colon V(G) \to \mathbb{Z}_2$. Then the following are equivalent:
		\begin{enumerate}
			\item $G_0 \cong G_{U}$,
			\item $\sum_{v \in V(G)} U(v) = 0$,
			\item $\hom(G, G_0) = \hom(G, G_{U})$.
		\end{enumerate}
	\end{lemma}
		
	\section{Decidability}
	\label{sec:decidable}
	
	As a first step, we show that the problem \HomIndP{$\mathcal{F}$} is decidable for every $k$-recognisable graph class $\mathcal{F}$ of treewidth at most $k-1$.
	We do so by establishing a bound on the maximum size of a graph $F \in \mathcal{F}$ for which $\hom(F, G) = \hom(F, H)$ needs to be checked in order to conclude whether $G \equiv_{\mathcal{F}} H$.

	For a graph class $\mathcal{F}$ and $\ell \in \mathbb{N}$, define the class $\mathcal{F}_{\leq \ell} \coloneqq \{F \in \mathcal{F} \mid |V(F)| \leq \ell\}$.

	\begin{theorem} \label{thm:decidable}
		Let $k \geq 1$.
		Let $\mathcal{F}$ be a graph class of treewidth $\leq k-1$ with $k$-recognisability index $C$. 
		For graphs $G$ and $H$ on at most $n$ vertices, 
		\[
			G \equiv_{\mathcal{F}} H \iff G \equiv_{\mathcal{F}_{\leq f(n, k, C)}} H
		\]
		where $f(n, k, C) \coloneqq \max\{ k^{2Cn^k}, 2Cn^k \}$.
	\end{theorem}

	Fix throughout a graph class $\mathcal{F}$ as in \cref{thm:decidable}.
	In reminiscence of Courcelle's theorem, 
	we let $Q$ denote the set of equivalence classes of~$\sim_{\mathcal{F}}$, as defined in \cref{def:recog}, and call them \emph{states}.
	A state $q \in Q$ is \emph{accepting} if for a (and equivalently, every) $\boldsymbol{F}$ in $q$ it holds that $\soe \boldsymbol{F} \in \mathcal{F}$.
	Write $A \subseteq Q$ for the set of all accepting states.
	
	The states are fine enough to distinguish between graphs in $\mathcal{F}$ and those not in $\mathcal{F}$. 
	Furthermore, they are preserved under series composition and gluing, which are the operations under which all graphs in $\mathcal{TW}(k)$ can be generated by \cref{lem:twk-gen}.
	
	\begin{lemma} \label{lem:fefvau}
		For $\boldsymbol{F}, \boldsymbol{F}', \boldsymbol{F}_1, \boldsymbol{F}_2, \boldsymbol{F}'_1, \boldsymbol{F}'_2 \in \mathcal{D}(k)$, $\boldsymbol{L} \in \mathcal{D}(k,k)$,
		\begin{enumerate}
			\item if $\boldsymbol{F}_1 \sim_{\mathcal{F}} \boldsymbol{F}'_1$ and $\boldsymbol{F}_2 \sim_{\mathcal{F}} \boldsymbol{F}'_2$ 
			then $\boldsymbol{F}_1 \odot \boldsymbol{F}_2 \sim_{\mathcal{F}} \boldsymbol{F}'_1 \odot \boldsymbol{F}'_2$,
			\item if $\boldsymbol{F} \sim_{\mathcal{F}} \boldsymbol{F}'$ then $\boldsymbol{L} \cdot \boldsymbol{F} \sim_{\mathcal{F}} \boldsymbol{L} \cdot \boldsymbol{F}'$.
		\end{enumerate}
	\end{lemma}
	\begin{proof}
		Let $\boldsymbol{K} = (K, \boldsymbol{u}) \in \mathcal{D}(k)$ be arbitrary. 
		Then $\soe((\boldsymbol{K} \odot \boldsymbol{F}_1) \odot \boldsymbol{F}_2) \in \mathcal{F} \Leftrightarrow \soe((\boldsymbol{K} \odot \boldsymbol{F}_1) \odot \boldsymbol{F}'_2) \in \mathcal{F}
		\Leftrightarrow \soe(\boldsymbol{K} \odot \boldsymbol{F}'_1 \odot \boldsymbol{F}'_2) \in \mathcal{F}$.
		
		For a $(k,k)$-bilabelled graph $\boldsymbol{L} = (L, \boldsymbol{u}, \boldsymbol{v})$, write $\boldsymbol{L}^* \coloneqq (L, \boldsymbol{v}, \boldsymbol{u})$ for the $(k,k)$-bilabelled graph obtained by swapping the in- and out-labels.
		Then $\soe(\boldsymbol{K} \odot (\boldsymbol{L} \cdot \boldsymbol{F})) = \soe( (\boldsymbol{L}^* \cdot \boldsymbol{K}) \odot \boldsymbol{F})$.
		Thus the second claim follows from the first.
	\end{proof}
	
	The key observation is that since $\sim_{\mathcal{F}}$ is preserved under these operations (\cref{lem:fefvau}) we can keep track of homomorphism counts of the graphs in each of the states using finite-dimensional vector spaces.
		
	Suppose wlog that $V(G)$ and $V(H)$ are disjoint.
	To every state $q \in Q$, 
	we associate vector spaces $S_d(q) \subseteq \mathbb{R}^{V(G)^k \cup V(H)^k}$ for $d \geq 1$.
	Recall that for a distinctly $k$-labelled graph $\boldsymbol{F} \in \mathcal{D}(k)$, 
	the vectors $\boldsymbol{F}_G \in \mathbb{R}^{V(G)^k}$ and $\boldsymbol{F}_H \in \mathbb{R}^{V(H)^k}$ are the homomorphism tensors of $\boldsymbol{F}$ w.r.t.\@ $G$ and $H$.
	Let $\boldsymbol{F}_G \oplus \boldsymbol{F}_H \coloneqq \left(\begin{smallmatrix}
		\boldsymbol{F}_G \\
		\boldsymbol{F}_H
	\end{smallmatrix} \right) \in \mathbb{R}^{V(G)^k \cup V(H)^k}$ denote the vector obtained by appending $\boldsymbol{F}_H$ to $\boldsymbol{F}_G$
	and define
	\[
		S_d(q) \coloneqq \left< \{ \boldsymbol{F}_G \oplus \boldsymbol{F}_H \mid \boldsymbol{F} \in \mathcal{TW}_d(k) \text{ in } q \} \right> \subseteq \mathbb{R}^{V(G)^k \cup V(H)^k},
	\]
	i.e.\ $S_d(q)$ is the subspace of $\mathbb{R}^{V(G)^k \cup V(H)^k}$ spanned by the stacked homomorphism tensors $\boldsymbol{F}_G \oplus \boldsymbol{F}_H$ for $\boldsymbol{F} \in \mathcal{TW}_d(k)$ in state~$q$.
	Clearly,  $S_d(q)$ is a subspace of $S_{d+1}(q)$ for every $d \geq 1$.
	Finally, let $S(q) \coloneqq \bigcup_{d\geq 1} S_d(q)$.

	Let $\mathcal{F}_{\geq k} \coloneqq \{F \in \mathcal{F} \mid \lvert V(F) \rvert \geq k\}$
	and write $\boldsymbol{1}_G$, $\boldsymbol{1}_H$ for the vectors in $\mathbb{R}^{V(G)^k \cup V(H)^k}$ which are one on $V(G)^k$ and $V(H)^k$, respectively, and zero everywhere else.
	The following \cref{lem:readout} describes how to recover homomorphism counts from the vectors in the~$S(q)$.
	
	\begin{lemma} \label{lem:readout}
		Two graph $G$ and $H$ are homomorphism indistinguishable over~$\mathcal{F}_{\geq k}$
		if and only if 
		for every $q \in A$ and every $v \in S(q)$ it holds that $\boldsymbol{1}_G^T v = \boldsymbol{1}_H^T v$. 
	\end{lemma}
	\begin{proof}
		For the forward direction, note that $S(q)$ is spanned by the  $\boldsymbol{F}_G \oplus \boldsymbol{F}_H$ where $\boldsymbol{F} \in \mathcal{TW}(k)$ is in $q$.
		Let $\boldsymbol{F}$ in $q \in A$ be arbitrary.
		Then $\soe \boldsymbol{F} \eqqcolon F \in \mathcal{F}_{\geq k}$ and it holds that $ \boldsymbol{1}_G^T (\boldsymbol{F}_G \oplus \boldsymbol{F}_H) = \soe(\boldsymbol{F}_G) = \hom(F, G) = \hom(F, H) = \boldsymbol{1}_H^T (\boldsymbol{F}_G \oplus \boldsymbol{F}_H)$.
		Conversely, let $F \in \mathcal{F}_{\geq k}$ be arbitrary. Since $\tw F \leq k-1$, by \cref{lem:treedec1}, there exists $\boldsymbol{u} \in V(F)^k$ such that $\boldsymbol{F} \coloneqq (F, \boldsymbol{u}) \in \mathcal{TW}(k)$.
		Furthermore, $\boldsymbol{F}$ belongs to some accepting~$q \in Q$. 
		Thus, $\boldsymbol{F}_G \oplus \boldsymbol{F}_H \in S(q)$, and hence
		$\hom(F, G) = \boldsymbol{1}_G^T (\boldsymbol{F}_G \oplus \boldsymbol{F}_H) = \boldsymbol{1}_H^T (\boldsymbol{F}_G \oplus \boldsymbol{F}_H) = \hom(F, H)$.
	\end{proof}
	
	The bound in \cref{thm:decidable} is essentially due to a dimension argument.
	Each of the spaces $S(q) \subseteq \mathbb{R}^{V(G)^k \cup V(H)^k}$ is of dimension at most $2n^k$.
Thus, the following \cref{lem:chain} yields that $S_d(q) = S(q)$ for $d \coloneqq 2Cn^k$ and all $q \in Q$.
	
	\begin{lemma} \label{lem:chain}
		If $d \geq 1$ is such that $S_d(q) = S_{d+1}(q)$ for all $q \in Q$
		then $S_d(q) = S(q)$ for all $q \in Q$.
	\end{lemma}
	\begin{proof}
		We argue that $S_d(q) \subseteq S_{d+i}(q)$ for all $i \geq 1$ by induction on $i$. The base case holds by assumption.
		The space $S_{d+i+1}(q)$ is spanned by the $\boldsymbol{F}_G \oplus \boldsymbol{F}_H$ where $\boldsymbol{F} \in \mathcal{TW}_{d+i+1}(k)$ 
		is in $q$.
		For such $\boldsymbol{F}$, by \cref{lem:twk-gen}, there exist $A \subseteq \binom{[k]}{2}$, $L \subseteq [k]$,
		and $\boldsymbol{F}^\ell \in \mathcal{TW}_{d+i}(k)$ for $\ell \in L$ such that 
		\[
			\boldsymbol{F} = \prod_{ij \in A} \boldsymbol{A}^{ij} \cdot  \bigodot_{\ell \in L} \boldsymbol{J}^\ell \boldsymbol{F}^\ell.
		\]
		Let $q_\ell$ denote the state of $\boldsymbol{F}^\ell$.
		By assumption, there exist $\boldsymbol{K}^{\ell m} \in \mathcal{TW}_d(k)$ in state $q_\ell$  and $\alpha_m \in \mathbb{R}$ such that $\boldsymbol{F}^\ell_G \oplus \boldsymbol{F}^\ell_H = \sum \alpha_m \boldsymbol{K}^{\ell m}_G \oplus \boldsymbol{K}^{\ell m}_H$.
		By \cref{lem:fefvau},
		\[
			\boldsymbol{F} \sim_{\mathcal{F}} \prod_{ij \in A} \boldsymbol{A}^{ij}  \cdot  \bigodot_{\ell \in L} \boldsymbol{J}^\ell \boldsymbol{K}^{\ell m}.
		\]
		for all $m$.
		Thus, $\boldsymbol{F}_G \oplus \boldsymbol{F}_H$ can be written as linear combination of vectors in $S_{d+i}(q) \subseteq S_d(q)$, by induction.
	\end{proof}

	This concludes the preparations for the proof of \cref{thm:decidable}.
	
	\begin{proof}[Proof of \cref{thm:decidable}]
		It suffices to prove the backward implication.
		Since $k \leq k^{2Cn^k}$, it suffices to show that $G \equiv_{\mathcal{F}_{\geq k}} H$ by verifying the condition in \cref{lem:readout}.
		By \cref{lem:chain}, $S_d(q) = S(q)$ for $d \coloneqq 2Cn^k$ and all $q \in Q$.
		Hence, $S(q)$ is spanned by the $\boldsymbol{F}_G \oplus \boldsymbol{F}_H$ where $\boldsymbol{F} \in \mathcal{TW}_d(k)$ is in state~$q$.
		By \cref{lem:tw-size}, these graphs has at most $\max\{k^{d}, d\} = \max\{  k^{2Cn^k}, 2Cn^k \}$ vertices.
		Thus,
		\(
			\boldsymbol{1}_G^T (\boldsymbol{F}_G \oplus \boldsymbol{F}_H) = \hom(\soe \boldsymbol{F}, G)
			= \hom(\soe \boldsymbol{F}, H)
			= \boldsymbol{1}_H^T (\boldsymbol{F}_G \oplus \boldsymbol{F}_H),
		\)
		as desired.
	\end{proof}

	Finally, we adapt the techniques developed above to prove the following analogue of \cref{thm:decidable} for graph classes of bounded pathwidth:
	
	\begin{theorem} \label{thm:decidable-pw}
		Let $k \geq 1$.
		Let $\mathcal{F}$ be a graph class of pathwidth $\leq k-1$ with $k$-recognisability index $C$. 
		For graphs $G$ and $H$ on at most $n$ vertices, 
		\[
		G \equiv_{\mathcal{F}} H \iff G \equiv_{\mathcal{F}_{\leq f(n, k, C)}} H
		\]
		where $f(n, k, C) \coloneqq 2Cn^k + k - 1$.
	\end{theorem}
	\begin{proof}
		We sketch how the proof of \cref{thm:decidable} needs to be adapted.
		Define $S_d(q)$ and $S(q)$ as above with $\mathcal{PW}_d(k)$ instead of $\mathcal{TW}_d(k)$, cf.\ \cref{sec:labelled-pw}.
		
		By \cref{lem:pwk-gen},
		Observing that for every $\boldsymbol{F} \in \mathcal{PW}_{d+1}(k)$ there exists $A \subseteq \binom{[k]}{2}$, $\ell \in [k]$, and $\boldsymbol{F}' \in \mathcal{PW}_{d}(k)$ such that
		\(
			\boldsymbol{F} =  \boldsymbol{A} \boldsymbol{J}^\ell \boldsymbol{F}'
		\)
		where $\boldsymbol{A} \coloneqq \prod_{ij \in A} \boldsymbol{A}^{ij}$, \cref{lem:chain} follows analogously.
		Thus, $S_d(q) = S(q)$ for all $q \in Q$ when $d \coloneqq 2Cn^k$.
		By \cref{lem:pwk-size}, this space is spanned by homomorphism tensors of graphs of size at most $2Cn^k +k -1$. 
		The statement follows as in \cref{thm:decidable}.
	\end{proof}

	Given \cref{thm:decidable,thm:decidable-pw},
	we suggest to consider the following general question:
	For a graph class $\mathcal{F}$, is there a function $f \colon \mathbb{N} \to \mathbb{N}$ such that $G \equiv_{\mathcal{F}} H$ if and only if $G \equiv_{\mathcal{F}_{\leq f(n)}} H$ for all graphs $G$ and $H$ on at most $n$ vertices?
	If such an $f$ exists, can be computed, and if membership in $\mathcal{F}$ is decidable then \HomIndP{$\mathcal{F}$} is decidable.
	
	\Cref{thm:decidable,thm:decidable-pw} give such functions for every recognisable graph class of bounded treewidth.
	By~\cite{lovasz_operations_1967}, for the class $\mathcal{G}$ of all graphs, $f$ can be taken to be the identity $n \mapsto n$.
	By~\cite{mancinska_quantum_2020}, there exists no such function for the class of all planar graphs $\mathcal{P}$ which is computable.

	\section{Modular Homomorphism Indistinguishability in Polynomial Time}
	\label{sec:ptime}

	The insight that yielded \cref{thm:decidable} is that the spaces $S(q)$ have polynomial dimension. 
	In this section, we strengthen this result by showing that bases $B(q)$ for the spaces $S(q)$ can be efficiently computed.
	A technical difficulty arising here is that the numbers produced in the process can be of doubly exponential magnitude. 
	In order to overcome this problem, we first consider the following problem:
	
	\dproblem{\textsc{ModHomInd}($\mathcal{F}$)}{Graphs $G$ and $H$, a prime $p$ in binary}{Are $G$ and $H$ homomorphism indistinguishable over $\mathcal{F}$ modulo~$p$?}
	
	It is easy to see that \cref{lem:readout,lem:chain} also hold when constructing the vector spaces over the finite field $\mathbb{F}_p$.
	Doing this has the advantage that all arithmetic operations can be carried out in $(\log p)^{O(1)}$.
	We prove the following:

	\begin{theorem} \label{thm:mod}
		Let $k \geq 1$.
		If $\mathcal{F}$ is a $k$-recognisable graph class of tree\-width $\leq k-1$
		then \textup{\textsc{ModHomInd}($\mathcal{F}$)} is in polynomial time.
	\end{theorem}

	For later reference, we recall the data describing the graph class~$\mathcal{F}$ which the algorithm utilises.
	Recall $Q$ the set of equivalence classes (\emph{states}) of $\sim_{\mathcal{F}}$ and $A \subseteq Q$ the set of accepting states from \cref{sec:decidable}.
	\Cref{lem:fefvau} asserts that the operations generating $\mathcal{TW}(k)$ can be regarded as operations on the equivalence classes of~$\sim_{\mathcal{F}}$.
	The algorithm is provided with a table encoding how each of the operation acts on the states.
	
	Write $\pi \colon \mathcal{TW}(k) \to Q$ for the map that associates an $\boldsymbol{F} \in \mathcal{TW}(k)$ to its state $q \in Q$.
	Write $q_0$ for the state of $\boldsymbol{1} \in \mathcal{TW}(k)$.
	Furthermore, write $g \colon Q \times Q \to Q$ and $b_{\boldsymbol{B}} \colon Q \to Q$ for every $\boldsymbol{B} \in \mathcal{B}(k)$
	such that
	\begin{align}
		g(\pi(\boldsymbol{F}), \pi(\boldsymbol{F}')) &= \pi(\boldsymbol{F} \odot \boldsymbol{F}'), \\
		b_{\boldsymbol{B}}(\pi(\boldsymbol{F})) &= \pi(\boldsymbol{B} \cdot \boldsymbol{F}).\label{eq:bb}
	\end{align}
	for every $\boldsymbol{F}, \boldsymbol{F}' \in \mathcal{TW}(k)$ and $\boldsymbol{B} \in \mathcal{B}(k)$.
	Note that $Q$, $A$, $g$, $q_0$ and the $b_{\boldsymbol{B}}$, $\boldsymbol{B} \in \mathcal{B}(k)$, are finite objects, which can be hard-coded. The map $\pi$ does not need to be computable and is only needed for analysing the algorithm.
	
	Finally, for a set $B \subseteq \mathbb{F}_p^n$, write $\left< B \right> \subseteq \mathbb{F}_p^n$ for the $\mathbb{F}_p$-vector space spanned by $B$, cf.\ \cref{check1,check2}.
	Having defined all necessary notions, we can now describe \Cref{alg1}.

	\begin{algorithm}
		\SetAlgoNoEnd
		\LinesNumbered
\KwIn{graphs $G$ and $H$, a prime $p$ in binary.}
		\KwData{$k$, $Q$, $A$, $q_0$, $g$, $b_{\boldsymbol{B}}$ for $\boldsymbol{B} \in \mathcal{B}(k)$.}
		\KwOut{whether $G \equiv_{\mathcal{F}}^p H$.}
		
		\nllabel{bruteforce} With brute force check whether $G$ and $H$ are homomorphism indistinguishable over the finite graph class $\mathcal{F}_{\leq k}$ modulo~$p$ and reject if not\;
		
		$B(q_0) \leftarrow \{\boldsymbol{1}_{G} \oplus \boldsymbol{1}_H\} \subseteq \mathbb{F}_p^{V(G)^k \cup V(H)^k}$\;
		
		$B(q) \leftarrow \emptyset \subseteq \mathbb{F}_p^{V(G)^k \cup V(H)^k}$ for all $q \neq q_0$\;

		\Repeat{none of the $B(q)$, $q \in Q$, are updated}{ \nllabel{mainloop}

			\ForEach{$\boldsymbol{B} \in \mathcal{B}(k)$, $q \in Q$, $v \in B(q)$}{\nllabel{subloop1}
				$w \leftarrow (\boldsymbol{B}_G \oplus \boldsymbol{B}_H) v \coloneqq \left( \begin{smallmatrix}			\boldsymbol{B}_G & 0 \\ 0 & \boldsymbol{B}_H			\end{smallmatrix}\right) v$\;
				
				\If{$w \not\in \left< B(b_{\boldsymbol{B}}(q)) \right>$}{\nllabel{check1} add $w$ to $B(b_{\boldsymbol{B}}(q))$\;}
			}

			\ForEach{$q_1,q_2 \in Q$, $v_1 \in B(q_1)$, $v_2 \in B(q_2)$}{\nllabel{subloop2}
				$w \leftarrow v_1 \odot v_2$\;	
				\If{$w \not\in \left< B(g(q_1, q_2)) \right>$}{\nllabel{check2} add $w$ to $B(g(q_1, q_2))$\;}
			}	
		}
		
		\eIf{\nllabel{eq:acc} $\boldsymbol{1}_G^T v = \boldsymbol{1}_H^Tv$ for all $q \in A$ and  $v \in B(q)$ }{accept\;}{reject\;}
		
		\caption{\textsc{ModHomInd}($\mathcal{F}$) for $k$-recognisable $\mathcal{F}$ of treewidth $\leq k-1$.}
		\label{alg1}
	\end{algorithm}
	
	Write $n \coloneqq \max \{|V(G)|, |V(H)|\}$ and $C \coloneqq |Q|$. The runtime of \cref{alg1} can be analysed as follows:
	
	\begin{lemma} \label{lem:runtime}
		There exists a computable function $f$ such that \Cref{alg1} runs in time 
		$f(k, C) n^{O(k)} (\log p)^{O(1)}$.
	\end{lemma}
	\begin{proof}
		Consider the following individual runtimes:
		
		Counting homomorphism of a graph on $k$ vertices into a graph on $n$ vertices, can be done in time $n^k$.
		Hence, \Cref{bruteforce} requires time $f(k) n^k$ for some computable $f$.
		
		Throughout the execution of \cref{alg1}, the vectors in each $B(q)$, $q \in Q$, are linearly independent. Thus, $|B(q)| \leq \dim S(q) \leq 2n^k$ and $\sum_{q \in Q} |B(q)| \leq  2Cn^k$.
		Hence, the body of the loop in \cref{mainloop}
		is entered at most $2Cn^k$ many times.
		
		The loop in \cref{subloop1} iterates over at most $k^2 \cdot C \cdot 2n^k$ many objects. 
		Computing the vector $w$ takes polynomial time in $2n^k \cdot \log p$.
		The same holds for checking the condition in \cref{check1,check2}, e.g.\ via Gaussian elimination.
		The loop in \cref{subloop2} iterates over at most $C^2 \cdot (2n^k)^2$ many objects.
		
		Finally, checking the condition in \cref{eq:acc} takes $C \cdot 2n^k \cdot (\log p)^{O(1)}$ many steps.
	\end{proof}

	The following \cref{lem:correct} yields that \cref{alg1} is correct.
	
	\begin{lemma} \label{lem:correct}
		When \Cref{alg1} terminates,
		$B(q)$ spans $S(q)$ for all $q \in Q$.
	\end{lemma}
	\begin{proof}
		First observe that the invariant $B(q) \subseteq S(q)$ for all $q \in Q$ is preserved throughout \cref{alg1}.
		Indeed, for example in \cref{subloop1}, since $v \in B(q) \subseteq S(q)$, 
		it can be written as linear combination of $\boldsymbol{F}_G \oplus \boldsymbol{F}_H$ for $\boldsymbol{F} \in \mathcal{TW}(k)$ of state $q$. 
		Because $\boldsymbol{B} \cdot \boldsymbol{F}$ is in state $b_{\boldsymbol{B}}(q)$ by \cref{eq:bb}, $(\boldsymbol{B}_G \oplus \boldsymbol{B}_H)v$ is in the span of $\boldsymbol{B}_G\boldsymbol{F}_G \oplus \boldsymbol{B}_H\boldsymbol{F}_H \in S(b_{\boldsymbol{B}}(q))$.
		
		Now consider the converse inclusion.
		The proof is by induction on the structure in \cref{lem:twk-gen}.
		By initialisation, $\boldsymbol{1}_G \oplus \boldsymbol{1}_H$ is in the span of $B(q_0)$.
		
		For the inductive step,
		suppose that $\boldsymbol{F} \in \mathcal{TW}(k)$ of state $q$ is such that 
		$\boldsymbol{F}_G \oplus \boldsymbol{F}_H = \sum_{v \in B(q)} \alpha_v v$ for some coefficients $\alpha_v \in \mathbb{F}_p$.
		Let $\boldsymbol{B} \in \mathcal{B}(k)$ and $\boldsymbol{F}' \coloneqq \boldsymbol{B} \cdot \boldsymbol{F}$.
		Then $(\boldsymbol{B}_G \oplus \boldsymbol{B}_H) v$
		is in the span of $B(b_{\boldsymbol{B}}(q))$ for all $v \in B(q)$
		by the termination condition.
		Hence, $\boldsymbol{F}'_G \oplus \boldsymbol{F}'_H = \sum_{v \in B(q)} \alpha_v (\boldsymbol{B}_G \oplus \boldsymbol{B}_H) v$ is in the span of $B(b_{\boldsymbol{B}}(q))$.
		
		Let $\boldsymbol{F}^1, \boldsymbol{F}^2 \in \mathcal{TW}(k)$ of states $q_1, q_2$ be such that 
		$\boldsymbol{F}^1_G \oplus \boldsymbol{F}^1_H = \sum_{v \in B(q_1)} \alpha_v v$ 
		and
		$\boldsymbol{F}^2_G \oplus \boldsymbol{F}^2_H = \sum_{w \in B(q_2)} \beta_w w$ 
		for some coefficients $\alpha_v,\beta_w \in \mathbb{F}_p$.
		Since the algorithm terminated, all $v \odot w$ for $v \in B(q_1)$ and $w \in B(q_2)$ are in the span of $B(g(q_1, q_2))$.
		Then $(\boldsymbol{F}^1 \odot \boldsymbol{F}^2)_G \oplus (\boldsymbol{F}^1 \odot \boldsymbol{F}^2)_H 
		= (\boldsymbol{F}^1_G \oplus \boldsymbol{F}^1_H) \odot (\boldsymbol{F}^2_G \oplus \boldsymbol{F}^2_H) 
		= \sum_{v \in B(q_1), w \in B(q_2)} \alpha_v \beta_w (v \odot w)$
		is in the span of $B(g(q_1, q_2))$.
	\end{proof}
	
	This concludes the preparations for the proof of \cref{thm:mod}:
	\begin{proof}[Proof of \cref{thm:mod}]
		\cref{lem:correct} implies that the conditions in \cref{lem:readout,eq:acc} are equivalent.
		Thus, $G$ and $H$ are homomorphism indistinguishable over $\mathcal{F}_{\geq k}$ modulo~$p$ if and only if \cref{eq:acc} holds.
		The runtime bound is given in \cref{lem:runtime}.
	\end{proof}

	\section{Polynomial Time for Bounded Pathwidth}

	In this section, we deduce \cref{thm:pw} from \cref{thm:mod,thm:decidable-pw}.
	The idea is to run \cref{alg1} for a collection of primes whose product is larger than any relevant homomorphism count and invoking the Chinese Remainder Theorem.
	The bound in \cref{thm:decidable-pw} is tight enough for this procedure to be in polynomial time.
	First consider the following \cref{lem:boundprime}.
	\begin{lemma} \label{lem:boundprime}
		Let $N, n \in \mathbb{N}$ be such that $N \log n \geq e^{2000}$.
		Let $\mathcal{F}$ be a graph class and $G$ and $H$ be graphs on at most $n$ vertices.
		If $G \not\equiv_{\mathcal{F}_{\leq N}} H$ then there exists a prime $N \log n < p \leq (N\log n)^2$ such that $G \not\equiv^p_{\mathcal{F}_{\leq N}} H$.
	\end{lemma}
	\begin{proof}
		By \cref{fact:primes},
		the product $M$ of the primes $N \log n \leq p < (N\log n)^2$ 
		is at least $2^{(N \log n)^2/2} > 2^{N \log n} = n^N$.
		If $G \not\equiv_{\mathcal{F}_{\leq N}} H$ then there exist $F \in \mathcal{F}_{\leq N}$ such that $\hom(F,G) \neq \hom(F, H)$.
		Both numbers are non-negative and at most $n^N$.
		Hence, $\hom(F, G) \not\equiv \hom(F, H) \mod M$.
		By the Chinese Reminder Theorem \cite[Theorem~2.10]{nathanson_elementary_2000}, there exists $N \log n < p \leq (N\log n)^2$ such that $\hom(F, G) \not\equiv \hom(F, H) \mod p$, as desired.
	\end{proof}

	We can now deduce \cref{thm:pw}:
	\mainPW*
	\begin{proof}
		Let $G$ and $H$ be graphs on at most $n$ vertices.
		By \cref{thm:decidable-pw}, $G \equiv_{\mathcal{F}} H$ if and only if $G \equiv_{\mathcal{F}_{\leq N}} H$ for $N \coloneqq 2Cn^k + k-1$.
		We may assume wlog that $N\log n \geq e^{2000}$.
		By \cref{lem:boundprime}, the latter holds if and only if $G \equiv^p_{\mathcal{F}_{\leq N}} H$ for all primes $N \log n < p \leq (N\log n)^2$.
		Moreover, $G \equiv_{\mathcal{F}} H$ implies $G \equiv^p_{\mathcal{F}} H$, which implies $G \equiv^p_{\mathcal{F}_{\leq N}} H$, for all primes~$p$.
		Hence, $G \equiv_{\mathcal{F}} H$ if and only if $G \equiv^p_{\mathcal{F}} H$ for all primes $N \log n < p \leq (N\log n)^2$.
		Thus, $G \equiv_{\mathcal{F}} H$ can be checked by enumerating all numbers between $N \log n$ and $(N\log n)^2$ and running \cref{alg1} for every prime~\cite{agrawal_primes_2004}.
	\end{proof}

	Similarly, \cref{lem:boundprime} places \HomIndP{$\mathcal{F}$} for a $k$-recognisable graph class $\mathcal{F}$ of treewidth at most $k-1$ in coNP. Indeed, by \cref{thm:decidable}, $N$ can be taken to be exponential in $n$. Thus, the prime from \cref{lem:boundprime} is a polynomial-size certificate for $G \not\equiv_{\mathcal{F}} H$ which can be verified in polynomial time by \cref{thm:mod}.

	\section{Randomised Polynomial Time}	
	
	The chance that a random prime $p$ of appropriate size certifies that $G \not\equiv_{\mathcal{F}} H$  as in \cref{lem:boundprime} is very high.
	This observation yields \cref{thm:coRP} via the following \cref{lem:prob-homind}:
	
	\begin{lemma} \label{lem:prob-homind}
		Let $N, n \in \mathbb{N}$ be such that $N \log n \geq e^{2000}$.
		Let $\mathcal{F}$ be a graph class and $G$ and $H$ be graphs on at most $n$ vertices.
		If $G \not\equiv_{\mathcal{F}_{\leq N}} H$
		then
		the probability that a random prime $N \log n < p \leq (N\log n)^2$
		is such that $G \equiv^p_{\mathcal{F}_{\leq N}} H$ is at most $\frac{2}{N \log n}$.
	\end{lemma}
	\begin{proof}
		By assumption, there exist $F \in \mathcal{F}_{\leq N}$ such that $\hom(F, G) \neq \hom(F, H)$.
		Both numbers are non-negative and at most $n^N$.
		Thus, $\hom(F, G) \not\equiv \hom(F, H) \mod n^N$.
		
		By the Chinese Reminder Theorem \cite[Theorem~2.10]{nathanson_elementary_2000}, in every set $X$  of at least $\frac{N \log n}{\log(N\log n)}$ many primes $N \log n < p \leq (N\log n)^2$, 
		there exists a $p \in X$ such that $\hom(F, G) \not\equiv \hom(F, H) \mod p$.
		This is because $\prod_{p \in X} p > (N \log n)^{\frac{N \log n}{\log(N\log n)}} = n^N$.
		
		By \cref{fact:primes}, the probability that a random prime $N \log n < p \leq (N\log n)^2$ is such that $\hom(F, G) \equiv\hom(F, H) \mod p$ is
		\[
			\leq \frac{N \log n}{\log(N \log n)} \cdot \frac{2\log(N \log n)}{(N \log n)^2} = \frac{2}{N \log n}.
		\]
		Since $G \equiv_{\mathcal{F}}^p H$ implies that $\hom(F, G) \equiv\hom(F, H) \mod p$, the desired probability is bounded by the same value.
	\end{proof}
	
	The algorithm that yields \cref{thm:coRP} can now be stated as \cref{alg2}.
	
	\mainRP*
	\begin{proof}
		Let $N \coloneqq \max\{ k^{2Cn^k}, 2Cn^k \}$ be the quantity from \cref{thm:decidable}. 
		We may suppose that $N \log n \geq e^{2000}$.
		Since $\log N$ is polynomial in the input size,
		\cref{alg2} runs in polynomial time, requiring a polynomial-time primality test~\cite{agrawal_primes_2004}.
		
		For correctness, first observe that if $G \equiv_{\mathcal{F}} H$ then \cref{alg1} always accepts.
		If $G \not\equiv_{\mathcal{F}} H$ then \cref{alg1} might accept incorrectly.
		In each iteration, the probability of not rejecting
		is the probability of not sampling a prime plus the probability of $p$ being such that $G \equiv_{\mathcal{F}}^p H$. 
		By \cref{fact:primes,lem:prob-homind},
		it is at most
		\[
			1 - \frac{1}{2 \log(N \log n)} + \frac{2}{N \log n} \leq 1 - \frac{1}{4 \log (N\log n)}.
		\]
		By \cite[4.2.30]{abramowitz_handbook_1965}, the total probability for accepting if $G \not\equiv_{\mathcal{F}} H$ is at most $e^{-1} < 1/2$.
	\end{proof}
	
	\begin{algorithm}
		\SetAlgoNoEnd
		\KwIn{graphs $G$ and $H$}
		\KwOut{whether $G \equiv_{\mathcal{F}} H$}
		
		$n \leftarrow \max\{|V(G)|, |V(H)|\}$\;
		
		$N \leftarrow \max\{ k^{2Cn^k}, 2Cn^k \}$\;
		
		\For{$\lceil 4\log(N\log n) \rceil$ times}{
		Sample a random number $N\log n < p < (N \log n)^2$\;
		\If{$p$ is a prime and $G \not\equiv_{\mathcal{F}}^p H$}{reject\;}
		}
		accept\;
		
		\caption{A randomised reduction from \HomIndP{$\mathcal{F}$} to \textsc{ModHomInd}($\mathcal{F}$) for $k$-recognisable graph classes $\mathcal{F}$ of treewidth $\leq k-1$.}
		\label{alg2}
	\end{algorithm}

	\section{Fixed-Parameter Tractability}
	\label{sec:alg-meta-thm}
	
	In this section, we deduce \cref{thm:fptRP} from \cref{thm:coRP}.
	The challenge is to efficiently compute from the \CMSO-sentence $\phi$ and $k$ the data describing the graph class $\mathcal{F}_{\phi, k}$ for \cref{alg1}.
	That this can be done was proven by Courcelle~\cite{courcelle_graph_2012}.
	
	More precisely, Courcelle proved that for every \CMSO-sentence~$\phi$ and integer $k$ one can compute a finite automaton processing expressions which encode (tree decompositions of) graphs of bounded treewidth. It is this automaton from which the data required by \cref{alg1} can be constructed. 
	In order to state this argument precisely with minimal technical overhead, we introduce the following syntactical counterpart of $\mathcal{TW}(k)$.
	
 	Write $\mathfrak{TW}(k)$ for the set of \emph{terms} defined inductively as the following formal expressions:
	\begin{enumerate}
		\item $\boldsymbol{1} \in \mathcal{TW}(k)$ is a term,
		\item if $t_1, t_2 \in \mathfrak{TW}(k)$ then $t_1 \odot t_2 \in \mathfrak{TW}(k)$,
		\item if $t \in \mathfrak{TW}(k)$ then $\boldsymbol{B} \cdot t \in \mathfrak{TW}(k)$ for every $\boldsymbol{B} \in \mathcal{B}(k)$.
	\end{enumerate}
	The key difference between $\mathcal{TW}(k)$ and $\mathfrak{TW}(k)$ is that for elements of the former the tree decomposition satisfying \cref{def:twk} is implicit.
	For applying Courcelle's techniques, we require this decomposition to be explicit, as in $\mathfrak{TW}(k)$.
	
	There is a mapping $\val \colon \mathfrak{TW}(k) \to \mathcal{TW}(k)$ interpreting the formal expressions above as concrete distinctly $k$-labelled graphs.
	By \cref{lem:twk-gen}, this mapping is surjective.	
	We require the following \cref{thm:courcelle}.
\begin{theorem}[Courcelle] \label{thm:courcelle}
		Given a \CMSO-sentence $\phi$ and an integer $k \in \mathbb{N}$,
		one can compute 
		a finite set $Q$, 
		a subset $A \subseteq Q$, 
		an element $q_0 \in Q$, and 
		functions $g \colon Q \times Q \to Q$
and $b_{\boldsymbol{B}} \colon Q \to Q$ for every $\boldsymbol{B} \in \mathcal{B}(k)$ 
		such that there exists a map $\pi \colon \mathfrak{TW}(k) \to Q$ satisfying
		\begin{enumerate}
			\item $\pi(\boldsymbol{1}) = q_0$,
			\item for $t \in \mathfrak{TW}(k)$,  $\soe(\val(t)) \models \phi$ iff $\pi(t) \in A$,
			\item for all $t_1, t_2 \in \mathfrak{TW}(k)$, $\pi(t_1 \odot t_2) = g(\pi(t_1), \pi(t_2))$,
			\item for all $t \in \mathfrak{TW}(k)$ and $\boldsymbol{B} \in \mathcal{B}(k)$, $\pi(\boldsymbol{B} \cdot t) = b_{\boldsymbol{B}}(\pi(t))$.
\end{enumerate}  
	\end{theorem}
\begin{proof}
	We sketch how the theorem follows from results in \cite{courcelle_graph_2012}.
	By the proof of \cite[Theorem~6.3]{courcelle_graph_2012}, 
	one can compute, given $\phi$ and $k$, a finite deterministic automaton 
	recognising whether a so called $F^{\textup{HR}}_{[k]}$-term 
	evaluates to a graph satisfying~$\phi$.
	
	$F^{\textup{HR}}_{[k]}$ plays the role of $\mathfrak{TW}(k)$. 
	It is a many-sorted signature \cite[Definition~2.123]{courcelle_graph_2012} whose sorts correspond 
	to the sources (labels, in our terminology) of sourced (distinctly labelled) graphs of bounded treewidth.
	The operations of $F^{\textup{HR}}_{[k]}$ are gluing, dropping of labels, and renaming of labels \cite[Definition~2.32]{courcelle_graph_2012}.
	Note that the series product, gluing, and unlabelling from \cref{sec:bilabelled} can be derived from these.
	Moreover, our graph $\boldsymbol{1} \in \mathcal{TW}(k)$ in \cref{fig:one} can be derived from the constants in~$F^{\textup{HR}}_{[k]}$.
	
	A finite deterministic automaton processing $F^{\textup{HR}}_{[k]}$-terms \cite[Defition~3.46]{courcelle_graph_2012}
	comprises a finite set of states $Q$, 
	a set of accepting states $A \subseteq Q$,
	and a transition function $\delta$ 
	mapping $(q_1, \dots, q_\ell, f)$ to $q$ where $q_1, \dots, q_\ell, q \in Q$
	and $f \in F^{\textup{HR}}_{[k]}$ is an $\ell$-ary operation, $\ell \geq 0$.
	Furthermore, there is  map $\sigma$ that associates to every $q \in Q$ a sort, that is a subset of $[k]$ of labels.
	
	The automaton processes $F^{\textup{HR}}_{[k]}$-terms, 
	which in our case are elements of $\mathfrak{TW}(k)$.
	In order to derive the desired objects from this automaton,
	we keep those states for which $\sigma$ evaluates to $[k]$.
	The maps $g$ and $b_{\boldsymbol{B}}$ can be computed from $\delta$.
	We mark a state $q$ as accepting if $\delta$ maps $q$ and $\soe$ (an operation derived from those in $F^{\textup{HR}}_{[k]}$) to an accepting state.
	The state $q_0$ is the state produced by $\delta$ on the $0$-ary derived function~$\boldsymbol{1}$.
	The map $\pi$ (which we do not require to be computable) is defined by the semantics of the automaton.
\end{proof}

	Equipped with \cref{thm:courcelle}, we prove \cref{thm:fptRP}.
	
	\mainFPT*
	\begin{proof}
		Let $\mathcal{F} \coloneqq \mathcal{F}_{\phi, k}$.
		Invoking \cref{thm:courcelle}, compute $Q$, $A$, $q_0$, $g$, and $b_{\boldsymbol{B}}$ for $\boldsymbol{B} \in \mathcal{B}(k)$ in time only depending on $\phi$ and $k$.
		Now invoke \cref{alg2}, which calls \cref{alg1}.
		
		Since the data obtained from \cref{thm:courcelle} 
		describes an equivalence relation $t_1 \sim t_2$ iff $\pi(t_1) = \pi(t_2)$ on $\mathfrak{TW}(k)$ rather than on $\mathcal{TW}(k)$ as in \cref{sec:ptime,sec:decidable},
		the definitions in these sections need to be slightly adapted.
		
		By \cref{thm:courcelle}, the following assertions analogous to \cref{lem:fefvau} hold: If $\pi(t_1) = \pi(t_2)$ and $\pi(t'_1) = \pi(t'_2)$  for $t_1, t_2, t'_1, t'_2 \in \mathfrak{TW}(k)$ then 
		$\pi(t_1 \odot t'_1) = g(\pi(t_1), \pi(t'_1)) = g(\pi(t_2), \pi(t'_2)) = \pi(t_1 \odot t'_1)$. Furthermore, for $\boldsymbol{B} \in \mathcal{B}(k)$, $\pi(\boldsymbol{B} \cdot t_1) = b_{\boldsymbol{B}}(\pi(t_1)) = b_{\boldsymbol{B}}(\pi(t_2)) = \pi(\boldsymbol{B} \cdot t_2)$.
		
		Define $S(q)$ for $q \in Q$ as the $\mathbb{F}_p$-vector space spanned by the $\boldsymbol{F}_G \oplus \boldsymbol{F}_H$ where $\boldsymbol{F} = \val(t)$ for some $t \in \mathfrak{TW}(k)$ in state~$q$.
		
		Then the following analogue of \cref{lem:readout} holds: Two graphs $G$ and $H$ are homomorphism indistinguishable over $\mathcal{F}_{\geq k}$ modulo~$p$ if and only if $\boldsymbol{1}_G^T v = \boldsymbol{1}_H^T v$ for every $v \in S(q)$ and $q \in A$.
		Indeed, if $F \in \mathcal{F}_{\geq k}$ then there
		there exist by \cref{thm:courcelle} a $t \in \mathfrak{TW}(k)$ such that $\soe(\boldsymbol{F}) \cong F$ for $\boldsymbol{F} \coloneqq \val(t)$.
		Then $\hom(F, G) \equiv \boldsymbol{1}_G^T(\boldsymbol{F}_G \oplus \boldsymbol{F}_H) \mod p$. 
		The claim follows as in \cref{lem:readout}.
		
		The proof of \cref{lem:correct} goes through analogously, replacing structural induction on $\mathcal{TW}(k)$ via \cref{lem:twk-gen} by structural induction on the definition of $\mathfrak{TW}(k)$.
		The same holds for \cref{lem:readout}.
		Finally, \cref{lem:runtime} gives the desired runtime.
	\end{proof}

	Analogously, the following \cref{thm:fptPW} can be derived from \cref{thm:pw}.
	The difference between the problems \textsc{PWHomInd} and \HomInd  is that treewidth is replaced by pathwidth.
	
	\pproblem{\textsc{PWHomInd}}{Graphs $G$, $H$, a $\mathsf{CMSO}_2$-sentence $\phi$, an integer $k$.}{$|\phi| + k$.}{Are $G$ and $H$ homomorphism indistinguishable over \( \mathcal{F}_{\phi,k} \coloneqq \{\text{graph }F \mid F \models \phi \text{ and } \pw F \leq k -1\}\)?}
	
	\Cref{thm:fptPW} places \textsc{PWHomInd} in the parametrised complexity class~XP.
	\begin{theorem} \label{thm:fptPW}
		There exists a computable function $f \colon \mathbb{N} \to \mathbb{N}$
		and a deterministic algorithm for \textup{\textsc{PWHomInd}} of runtime $f(|\phi| + k)n^{O(k)}$.
	\end{theorem}
	\begin{proof}
		The desired algorithm can be obtained as described in \cref{thm:fptRP} disregarding the gluing operation.
		By \cref{thm:decidable-pw}, the relevant numbers are not too large and all necessary modulo can be computed explicitly as in \cref{thm:pw}.
	\end{proof}

	\section{Lasserre in Polynomial Time}
	\label{sec:lasserre}
	
	By phrasing graph isomorphism as an integer program,
	heuristics from integer programming can be used to attempt to solve graph isomorphism.
	Prominent heuristics are the Sherali--Adams linear programming hierarchy and the Lasserre semidefinite programming hierarchy.
	While  the (approximate) feasibility of each level of these hierarchies can be decided efficiently, it is known that a linear number of levels is required to decide graph isomorphism for all graphs \cite{atserias_sherali-adams_2012,grohe_pebble_2015,roberson_lasserre_2023}.

	In \cite{roberson_lasserre_2023}, for each $t \geq 1$, 
	feasibility of the $t$-th level of the Lasserre hierarchy 
	was characterised as homomorphism indistinguishability relation over a graph class $\mathcal{L}_t$ constructed ibidem.
	Moreover, the authors asked whether there is a polynomial-time algorithm for deciding these relations. 
	In this section, we give a randomised algorithm for this problem which is polynomial-time for every level.
	
	The Lasserre semidefinite program can be solved approximately in polynomial time using e.g.\ the ellipsoid method. How to decide \emph{exact} feasibility is generally unknown.
	Since $\mathcal{L}_t$ is a minor-closed graph class of treewidth at most~$3t-1$,
	\Cref{thm:coRP} immediately yields a randomised polynomial-time algorithm for each level of the hierarchy. 
	However, it is not clear how to compute the data in \cref{thm:courcelle} for $\mathcal{L}_t$ given $t$.
	The following \cref{thm:lasserre} overcomes this problem by making the dependence on the parameter~$t$ effective:

	\dproblem{\textsc{Lasserre}}{Graphs $G$ and $H$, an integer $t \in \mathbb{N}$.}{Is the level-$t$ Lasserre relaxation for the standard integer program for $G \cong H$ feasible?}

	\begin{theorem} \label{thm:lasserre}
		There exists a computable function $f \colon \mathbb{N} \to \mathbb{N}$ and a randomised algorithm deciding \textup{\textsc{Lasserre}} which always runs in time~$f(t)n^{O(t)}$, accepts all \textsmaller{YES}-instances and accepts \textsmaller{NO}-instances with probability less than one half.
	\end{theorem}
	
	By \cite{roberson_lasserre_2023}, deciding \textsc{Lasserre} amounts to deciding homomorphism indistinguishability over the graph class $\mathcal{L}_t$ defined ibidem.
	We recall the definition:
	
	A $(t,t)$-bilabelled graph $\boldsymbol{A} = (A, \boldsymbol{u}, \boldsymbol{v})$ is \emph{atomic} if all its vertices are labelled, i.e.\ every $v \in V(A)$ appears among the labelled vertices $u_1, \dots, u_t,\allowbreak  v_1, \dots, v_t$.
	Let $\mathcal{A}_t$ denote the set of all atomic graphs.
	The class $\mathcal{L}_t$ is the closure of $\mathcal{A}_t$ under series composition, parallel composition with atomic graphs, and arbitrary permutations of the labels, i.e.\ the action of the symmetric group~$\mathfrak{S}_{2t}$ on the $2t$ labels.
	
	As in \cref{sec:alg-meta-thm}, write $\mathfrak{L}_t$ for the following inductively defined set of formal expressions. We simultaneously define a complexity measure $\depth \colon \mathfrak{L}_t \to \mathbb{N}$ on this set:
	\begin{enumerate}
		\item if $\boldsymbol{A}$ is atomic then $\boldsymbol{A} \in \mathfrak{L}_t$ and $\depth \boldsymbol{A} \coloneqq 1$,
		\item if $w \in \mathfrak{L}_t$ and $\boldsymbol{A}$ is atomic then $\boldsymbol{A} \odot w \in \mathfrak{L}_t$ and $\depth(\boldsymbol{A} \odot w) \coloneqq \depth w$,
		\item if $ w \in \mathfrak{L}_t$ and $\sigma \in \mathfrak{S}_{2t}$ then $w^\sigma \in \mathfrak{L}_t$ and $\depth(w^\sigma) \coloneqq \depth w$,
		\item if $w_1, w_2 \in \mathfrak{L}_t$ then $w_1 \cdot w_2 \in \mathfrak{L}_t$ and $\depth(w_1 \cdot w_2) \coloneqq \max\{\depth w_1, \depth w_2\} + 1$.
	\end{enumerate}
	Clearly, each expression $w \in \mathfrak{L}_t$ can be associated with the $(t,t)$-bilabelled graph $\val(w) \in \mathcal{L}_t$ it encodes. 
	
	As in \cref{sec:decidable}, we first give an upper bound on the size of the graphs whose homomorphism counts need to be considered in order to decide \textsc{Lasserre}.
	
	\begin{theorem} \label{thm:decidable-lasserre}
		For $t \geq 1$ and graphs $G$ and $H$ on at most $n$~vertices, 
		\[
		G \equiv_{\mathcal{L}_t} H \iff G \equiv_{(\mathcal{L}_t)_{\leq f(n, t)}} H 
		\]
		where $f(n, t) \coloneqq 2t \cdot 4^{n^{2t}}$.
	\end{theorem}

	Towards \cref{thm:decidable-lasserre}, we make the following observations:
	
	\begin{lemma}
		\label{lem:lass-size}
		Let $t \geq 1$. 
		If $w \in \mathfrak{L}_t$ then $\val(w)$ is a bilabelled graph on at most $2t \cdot 2^{\depth(w)}$ vertices.
	\end{lemma}
	\begin{proof}
		The proof is by induction on the definition of $\mathfrak{L}_t$.
		If $w$ is atomic than $\val(w)$ has at most $2t$ vertices and $\depth(w) =1$.
		
		If $w = \boldsymbol{A} \odot w'$ for some atomic $\boldsymbol{A}$ then $\val(w) = \boldsymbol{A} \odot \val(w')$ has at most as many vertices as $\val(w')$ since all vertices of $\boldsymbol{A}$ are labelled. This implies the claim.
		The case $w = (w')^\sigma$ for $\sigma \in \mathfrak{S}_{2t}$ is analogous.
		
		If $w = w_1 \cdot w_2$ then the number of vertices in $\val(w)$ is at most the number of vertices in $\val(w_1)$ plus the number of vertices in $\val(w_2)$. This number is by induction at most $2t \cdot (2^{\depth w_1} + 2^{\depth w_2}) \leq 2t \cdot 2 \cdot 2^{(\depth w) -1} = 2t \cdot 2^{\depth w}$, as desired.
	\end{proof}
	
	For two graphs $G$ and $H$ on at most $n$ vertices,
	define $S_d$ for $d \geq 1$ as the subspace of $\mathbb{R}^{(V(G)^t \cup V(H)^t) \times (V(G)^t \cup V(H)^t)}$ spanned by $\boldsymbol{L}_G \oplus \boldsymbol{L}_H \coloneqq \left( \begin{smallmatrix}
		\boldsymbol{L}_G & 0 \\ 0 & \boldsymbol{L}_H
	\end{smallmatrix}\right)$ where $\boldsymbol{L} = \val(w)$ for $w \in \mathfrak{L}_t$ with $\depth(w) \leq d$.
	The dimension of $S \coloneqq \bigcup_{d \geq 1} S_d$ is at most $2n^{2t}$.
	Thus, by the following \cref{lem:lass-chain}, $S_d = S$ for $d \coloneqq 2n^{2t}$.

	\begin{lemma}
		\label{lem:lass-chain}
		If $d \geq 1$ is such that $S_d = S_{d+1}$ then $S_d = S$.
	\end{lemma}
\begin{proof}We show by induction that $S_{d+i} \subseteq S_d$ for all $i \geq 1$.
	The base case follows by assumption.
	For the inductive step, we argue by structural induction that for every $w \in \mathfrak{L}_t$ with $\depth(w) \leq d+i+1$ and $\boldsymbol{L} \coloneqq \val(w)$ it holds that $\boldsymbol{L}_G \oplus \boldsymbol{L}_H \in S_d$.
	If $w$ is atomic, then this clearly holds.
	If $w = \boldsymbol{A} \odot w'$ then by the inner induction the homomorphism tensors of $\val(w')$ is in $S_d$. 
	The Schur product of any matrix in this space with $\boldsymbol{A}_G \oplus \boldsymbol{A}_H$ is in $S_d$ by definition. Thus, the claim follows.
	The case $w = (w')^\sigma$ for $\sigma \in \mathfrak{S}_{2t}$ is analogous.
	If $w = w' \cdot w'$ then $\depth(w'), \depth(w'') < \depth(w) \leq d+i+1$. By the outer inductive hypothesis, $S_{d+i} \subseteq S_d$.
	Thus, writing $\boldsymbol{L}' \coloneqq \val(w')$ and $\boldsymbol{L}'' \coloneqq \val(w'')$, it holds that $\boldsymbol{L}'_G \oplus \boldsymbol{L}'_H, \boldsymbol{L}''_G \oplus \boldsymbol{L}''_H \in S_d$.
	Thus, writing $\boldsymbol{L} \coloneq \val(w)$, $\boldsymbol{L}_G \oplus \boldsymbol{L}_H = (\boldsymbol{L}' \boldsymbol{L}'')_G \oplus (\boldsymbol{L}'\boldsymbol{L}'')_H = (\boldsymbol{L}'_G \oplus \boldsymbol{L}'_H)(\boldsymbol{L}''_G \oplus \boldsymbol{L}''_H)\in S_{d+1} \subseteq S_d$.
\end{proof}

	This concludes the preparations for the proof of \cref{thm:decidable-lasserre}.
	
	\begin{proof}[Proof of \cref{thm:decidable-lasserre}]
		By \cref{lem:lass-chain}, $S_{d} =S$ for $d \coloneqq 2n^{2t}$.
		The space $S_d$ is spanned by $\boldsymbol{L}^i_G \oplus \boldsymbol{L}^i_H$ for some $\boldsymbol{L}^i \in \mathcal{L}_t$ on at most $2t \cdot 2^d$ vertices, by \cref{lem:lass-size}.
		For arbitrary $\boldsymbol{F} \in \mathcal{L}_t$, there exist coefficients $\alpha_i \in \mathbb{R}$ such that $\boldsymbol{F}_G \oplus \boldsymbol{F}_H = \sum \alpha_i \boldsymbol{L}^i_G \oplus \boldsymbol{L}^i_H$. Hence,
		\begin{align*}
			\hom(\soe \boldsymbol{F}, G) 
			&= \boldsymbol{1}_G^T (\boldsymbol{F}_G \oplus \boldsymbol{F}_H) \boldsymbol{1}_G
			= \sum \alpha_i \boldsymbol{1}_G^T (\boldsymbol{L}^i_G \oplus \boldsymbol{L}^i_H) \boldsymbol{1}_G \\
			&= \sum \alpha_i  \hom(\soe \boldsymbol{L}^i,G)
			= \hom(\soe \boldsymbol{F}, H), 
		\end{align*}
		as desired.
	\end{proof}
	
	We now describe \cref{alg3} for \textsc{Lasserre} and prove \cref{thm:lasserre}.

	\begin{algorithm}
	\SetAlgoNoEnd
\KwIn{graphs $G$ and $H$, $t \in \mathbb{N}$, a prime $p$ in binary.}
	\KwOut{whether $G \equiv_{\mathcal{L}_t}^p H$.}

	$B \leftarrow \{ \boldsymbol{A}_G \oplus \boldsymbol{A}_H \mid \boldsymbol{A} \in \mathcal{A}_t\} \subseteq \mathbb{F}_p^{(V(G)^t \cup V(H)^t) \times (V(G)^t \cup V(H)^t)}$\;	
	
	\Repeat{$B$ is not updated}{

		\ForEach{$\boldsymbol{A} \in \mathcal{A}_t$, $v \in B$}{
			$w \leftarrow (\boldsymbol{A}_G \oplus \boldsymbol{A}_H) \odot v$\;
			
			\If{$w \not\in \left< B \right>$}{ add $w$ to $B$\;}
		}

		\ForEach{$v_1, v_2 \in B$}{
			$w \leftarrow v_1 \cdot v_2$\;	
			\If{$w \not\in \left<B\right>$}{add $w$ to $B$\;}
		}
		\ForEach{$1 \leq i < j \leq 2t$, $v \in B$}{
			$w \leftarrow v^{(i \ j )}$\;	
			\If{$w \not\in \left<B\right>$}{add $w$ to $B$\;}
		}	
	}
	
	\eIf{$\boldsymbol{1}_G^T v \boldsymbol{1}_G = \boldsymbol{1}_H^Tv \boldsymbol{1}_H$ for all $v \in B$ }{accept\;}{reject\;}
	
	\caption{Modular \textsc{Lasserre}.}
	\label{alg3}
\end{algorithm}

	\begin{proof}[Proof of \cref{thm:lasserre}]
		Correctness follows as in \cref{lem:correct} for \cref{alg1}.
		For the runtime, note that the main loop is entered at most $2n^{2t}$ times. The subloops are entered $4t(2t-1) 2n^{2t} + 4n^{4t}$ times.
		The linear algebraic operations can be performed in time polynomial in $2n^{2t} \log p$. 
		Thus the algorithm has runtime $O(t^2 n^{Ct} (\log p)^C)$ for some constant $C \geq 4$.
		
		To close the gap between modular homomorphism indistinguishability over $\mathcal{L}_t$ and \textsc{Lasserre}, invoke \cref{lem:prob-homind,thm:decidable-lasserre}.
		\Cref{alg2} yields the desired randomised algorithm with $\mathcal{F} \coloneqq \mathcal{L}_t$ and adjusting $N \leftarrow 2t \cdot 4^{n^{2t}}$ in the second line.
	\end{proof}

	\section{Lower Bounds}
	
	In this final section, we establish two hardness results for the problem \HomInd.
	In both cases, we show hardness for natural families of graph classes.
	The approaches are orthogonal in the sense that reduction yielding coNP-hardness in \cref{thm:main3} is from a fixed-parameter tractable problem while the reduction yielding coW[1]-hardness in \cref{thm:main4} is not polynomial time.

	\subsection{coNP-Hardness}
	\label{sec:coNP}

	We show that the following problem is coNP-hard.
	\dproblem{WL}{Graphs $G$ and $H$, an integer $k$}{Are $G$ and $H$ $k$-WL indistinguishable?}

	By \cite{dvorak_recognizing_2010,cai_optimal_1992}, 
	$k$-WL indistinguishability coincides with homomorphism indistinguishability over the class of graphs of treewidth at most $k$.
	Hence, \textsc{WL} is clearly a special case of \HomInd, i.e.\@ with $\phi$ set to true.
	Thus, the following \cref{thm:main3} shows that when disregarding the parametrisation \HomInd is coNP-hard under polynomial-time many-one reductions.
	
	\mainC*
	\begin{proof}
		For a graph $G$,
		write $\Delta(G)$ for its maximum vertex degree.
		The following problem is NP-complete by \cite[Theorem~11]{bodlaender_treewidth_1997}:
		
		\dproblem{BoundedDegreeTreewidth}{a graph $G$ with $\Delta(G) \leq 9$, an integer $k$}{Is $\tw G \leq k$?}
		
		By deleting isolated vertices, we may suppose that every connected component of $G$ contains at least two vertices.		
		If $G$ has multiple connected components, 
		take one vertex from each component and connect them in a pathlike fashion. 
		This increases the maximum degree potentially by one but makes the graph connected. 
		The treewidth is invariant under this operation.
		Thus, we may suppose that $G$ is connected and $\Delta(G) \leq 10$.
		
		Given such an instance, we produce the instance $(G_0, G_1, k)$ of \textsc{WL}.
		Here, $G_0$ and $G_1$ are the even and odd \textsmaller{CFI} graphs as defined in \cref{sec:cfi}.
		
		Then $G_0$ and $G_1$ are $k$-WL indistinguishable if and only if $\tw G \geq k+1$.
		Indeed,  by \cite{cai_optimal_1992,dvorak_recognizing_2010} and \cite[Lemma~4.4]{neuen_homomorphism-distinguishing_2023}, 
		since $G$ is connected, if $\tw G \geq k+1$ then $G_0$ and $G_1$ are $k$-WL indistinguishable. 
		Conversely, if $\tw G < k+1$ then $G_0$ and $G_1$ are distinguished by $k$-WL since $\hom(G, G_0) \neq \hom(G, G_1)$ by \cref{lem:cfi} and \cite{dvorak_recognizing_2010}.
				
		Hence, $(G, k)$ is a \textsmaller{YES}-instance of \textsc{BoundedDegree\-Tree\-width}
		if and only if $(G_0, G_1, k)$ is a \textsmaller{NO}-instance of \textsc{WL}.
		The graphs $G_0$ and $G_1$ are of size $\sum_{v\in V(G)} 2^{\deg(v) - 1} \leq 2^9 n$, which is polynomial in the input.
	\end{proof}
	
	\subsection{coW[1]-Hardness}
	
	The second hardness result concerns \HomInd as a parametrised problem.
	Write $\mathcal{G}_{\leq k}$ for the class of all graphs on at most $k$ vertices and consider the following problem:
	
	\pproblem{HomIndSize}{Graphs $G$ and $H$, an integer $k \geq 1$.}{$k$}{Are $G$ and $H$ homomorphism indistinguishable over the class $\mathcal{G}_{\leq k}$?}
	
	The problem \textsc{HomIndSize} fixed-parameter reduces to \HomInd and also to \textsc{PWHomInd}.
	To that end, consider the first-order formula \[\phi_k \coloneqq  \exists x_1 \dots \exists x_k \forall y \bigvee_{i=1}^k (y = x_i) \] for $k \in \mathbb{N}$.
	Then, a graph models $\phi_k$ if and only if it has at most~$k$ vertices.
	Furthermore, $|\phi_k| = O(k)$. Hence, transforming the instance $(G, H, k)$ of \textsc{HomIndSize} to the instance $(G, H, \phi_k, k-1)$ of \HomInd gives the desired reduction.
	Since $|\phi_k| + k = O(k)$, \cref{thm:homindsize} implies \cref{thm:main4}.

	\begin{theorem} \label{thm:homindsize}
		\textup{\textsc{HomIndSize}} is \textup{coW[1]}-hard under fpt-reductions.
		Unless \textsmaller{ETH} fails, there is no algorithm for 	\textup{\textsc{HomIndSize}} that runs in time $f(k)n^{o(k)}$ for any computable function $f \colon \mathbb{N} \to \mathbb{N}$.
	\end{theorem}
	
	\begin{proof}
The proof is by reduction from the parametrised clique problem \textsc{Clique}, 
		which is well-known to be W[1]-complete and which does not admit
		an $f(k)n^{o(k)}$-time algorithm for any computable $f$ unless \textsmaller{ETH} fails \cite[Theorems~13.25, 14.21]{cygan_parameterized_2015}.
		
		Let $K$ denote the $k$-vertex clique and $K_0$ and $K_1$ its even and odd \textsmaller{CFI} graphs
		as defined in \cref{sec:cfi}.
		We first observe that $K_0 \equiv_{\mathcal{G}_{\leq k} \setminus \{K\}} K_1$.
		Indeed, by \cite[Theorem~3.13]{roberson_oddomorphisms_2022},
		for every graph~$F$, $\hom(F, K_0) \neq \hom(F, K_1)$ if and only if
		there exists a \emph{weak oddomorphism} $h \colon F \to K$ as defined in \cite[Definition~3.9]{roberson_oddomorphisms_2022}.
		By definition, a weak oddomorphism is surjective on edges and vertices, i.e.\@ for every $uv \in E(K)$
		there exists $u'v' \in E(F)$ such that $h(u'v') = uv$. 
		Hence, if $\hom(F, K_0) \neq \hom(F, K_1)$ then $F$ has at least $k$ vertices and $\binom{k}{2}$ edges.
		The only graph in $\mathcal{G}_{\leq k}$ matching this description is~$K$.
		
		The reduction produces given the instance $(G, k)$ of \textsc{Clique}
		the instance $(G \times K_0, G \times K_1, k)$ of \textsc{HomIndSize}
		where $K$ is the $k$-vertex clique and $\times$ 
		denotes the categorical product of two graphs.
		This instance can be produced in fixed-parameter time.
		Furthermore, the parameter~$k$ is not affected by this reduction.
		For correctness, consider the following cases:
		
		If $G \times K_0 \equiv_{\mathcal{G}_{\leq k}} G \times K_1$ then $\hom(K, G) = 0$.
		Indeed, by assumption and \cite[(5.30)]{lovasz_large_2012}, 
		\[ \hom(K, G)\hom(K, K_0) =  \hom(K, G \times K_0) = \hom(K, G \times K_1) = \hom(K, G) \hom(K, K_1). \]
		However, $\hom(K,K_0) \neq \hom(K,K_1)$ by \cref{lem:cfi}, and thus $\hom(K,G) = 0$.
		
		Conversely, it holds that $K_0 \equiv_{\mathcal{G}_{\leq k} \setminus \{K\}} K_1$
		and hence also $G \times K_0 \equiv_{\mathcal{G}_{\leq k} \setminus \{K\}} G \times K_1$
		by the initial observation and \cite[(5.30)]{lovasz_large_2012}.
		Since $\hom(K, G) = 0$, also $G \times K_0 \equiv_{\mathcal{G}_{\leq k}} G \times K_1$.
	\end{proof}
	
	\subsection{Beyond Treewidth}
	
	\Cref{thm:coRP} establishes that \HomIndP{$\mathcal{F}$} is decidable for every recognisable graph class of bounded treewidth. 
	It would be desirable to give a necessary criterion for a graph class $\mathcal{F}$ to be such that \HomIndP{$\mathcal{F}$} is decidable/tractable.
	
	In this section, we note that the requirement of bounded treewidth can be relaxed when one does not aim at a polynomial-time algorithm.
	In the following \cref{thm:cliques}, $\chi(F)$ denotes the chromatic number of $F$. 
	A typical graph class to which the theorem applies is the class $\mathcal{K}$ of all cliques.
	
	\begin{theorem} \label{thm:cliques}
		Let $\phi$ be a \CMSO-sentence and $\mathcal{F}$ the class of all graphs satisfying $\phi$.
		Suppose there exists a computable non-decreasing function $g \colon \mathbb{N} \to \mathbb{N}$
		such that $\tw F \leq g(\chi(F))$ for all $F \in \mathcal{F}$.
		Then $\HomInd(\mathcal{F})$ is decidable.
	\end{theorem}
	\begin{proof}
		Let $n \coloneqq \max\{|V(G)|, |V(H)|\}$.
		Observe that $G \equiv_{\mathcal{F}} H$ if and only if $G \equiv_{\mathcal{F} \cap \mathcal{TW}_{g(n)}} H$.
		Indeed, if $F \in \mathcal{F}$ has treewidth greater than $g(n)$ then $\chi(F) > n$.
		Since $G$ and $H$ both have chromatic number at most $n$, there is neither a homomorphism $F \to G$ nor $F \to H$.
		Now invoke \cref{thm:fptRP} with input $\phi$ and $g(n)$.
	\end{proof}
	
	The class $\mathcal{K}$ is definable in first-order logic and has constant clique-width. Thus, from a model checking perspective, it is a rather simple graph class.
	However, as shown in~\cite{boeker_complexity_2019}, $\HomInd(\mathcal{K})$ is $\mathsf{C}_=\mathsf{P}$-hard and thus not in polynomial time unless the polynomial hierarchy collapses.
	This indicates that one cannot replace Courcelle's graph algebras for bounded treewidth, which were used in this work, by his graph algebras for bounded clique-width, which Courcelle uses to prove that $\mathsf{MSO}_1$-model checking is fixed-parameter tractable on classes of bounded clique-width.
	
	\section{Conclusion}
	
	\Cref{thm:coRP}, our central result, asserts that deciding homomorphism indistinguishable is tractable over every recognisable graph class of bounded treewidth.
	Thereby, we make the intuition precise that finite-dimensional spaces of homomorphism tensors of graphs with a bounded number of labels are sufficient to efficiently capture homomorphism indistinguishability.
	
	A reasonable next step is to combine \cref{thm:coRP} with a hardness result giving necessary conditions on a graph class $\mathcal{F}$ for \HomIndP{$\mathcal{F}$} to be undecidable.
	We propose the following working hypothesis:
	
	\begin{question} \label{question}
		Is it true that for a proper minor-closed graph class~$\mathcal{F}$
		the problem \HomIndP{$\mathcal{F}$} is decidable if and only if $\mathcal{F}$ has bounded treewidth?
	\end{question}
	\Cref{thm:coRP} establishes the backward implication.
	The only known minor-closed graph class $\mathcal{F}$ for which \HomIndP{$\mathcal{F}$} is undecidable, is the class of planar graphs, as shown by Man\v{c}inska and Roberson~\cite{mancinska_quantum_2020}.
	An affirmative answer to \cref{question} would imply a conjecture of Roberson~\cite[Conjecture~5]{roberson_oddomorphisms_2022} asserting that $\equiv_{\mathcal{F}}$ is not isomorphism for every proper minor- and union-closed class~$\mathcal{F}$.
	However,
	even though the graph classes in \cref{question} clearly need to satisfy some condition which rules out pathological behaviours, it is not clear whether closure under taking minors is the right one \cite{roberson_oddomorphisms_2022,seppelt_logical_2023}. 
	 
	Towards \cref{question}, one could try to devise reductions between \HomIndP{$\mathcal{F}_1$} and \HomIndP{$\mathcal{F}_2$} for distinct (minor-closed) graph classes $\mathcal{F}_1$ and $\mathcal{F}_2$. We are not aware of any such reduction.
	
	A final question with repercussions beyond the scope of this work concerns the parametrised hardness of deciding Weisfeiler--Leman indistinguishability. If $k$ is a made a parameter in the problem \textsc{WL} from \cref{sec:coNP}, can one show coW[1]-hardness or even XP-completeness, cf.\ \cite{berkholz_lower_2012}?

	\paragraph{Acknowledgements}
	I was supported by the German Research Foundation (\textsmaller{DFG}) within Research Training Group 2236/2 (\textsmaller{UnRAVeL}) and by the European Union (\textsmaller{ERC}, SymSim, 101054974).
	
	Views and opinions expressed are however those of the author(s) only and do not necessarily reflect those of the European Union or the European Research Council Executive Agency. Neither the European Union nor the granting authority can be held responsible for them.
	
	I would like to thank Martin Grohe and Louis Härtel for fruitful discussions and their support.
	In particular, I thank Martin Grohe for making me aware of Chinese remaindering.

	\bibliographystyle{plainurl}

\end{document}